\documentclass[prr,reprint,superscriptaddress]{revtex4-2}
\pdfoutput=1
\usepackage{braket,dsfont,subfigure,amsfonts,amssymb,bm,graphicx,float,amsmath,hyperref,amsthm,multirow}
\usepackage{mathrsfs, xcolor}
\newcommand{\Tr}{{\rm Tr}}

\theoremstyle{definition}
\newtheorem{definition}{Definition}
\newtheorem{theorem}{Theorem}
\begin{document}
\title{Composite Quantum Phases in Non-Hermitian Systems}
\author{Yuchen Guo}
\thanks{These authors contributed equally.}
\affiliation{State Key Laboratory of Low Dimensional Quantum Physics and Department of Physics, Tsinghua University, Beijing 100084, China}
\author{Ruohan Shen}
\thanks{These authors contributed equally.}
\affiliation{State Key Laboratory of Low Dimensional Quantum Physics and Department of Physics, Tsinghua University, Beijing 100084, China}
\author{Shuo Yang}
\email{shuoyang@tsinghua.edu.cn}
\affiliation{State Key Laboratory of Low Dimensional Quantum Physics and Department of Physics, Tsinghua University, Beijing 100084, China}
\affiliation{Frontier Science Center for Quantum Information, Beijing 100084, China}
\affiliation{Hefei National Laboratory, Hefei 230088, China}

\begin{abstract}
    Non-Hermitian systems have attracted considerable interest in recent years owing to their unique topological properties that are absent in Hermitian systems. 
    While such properties have been thoroughly characterized in free fermion models, they remain an open question for interacting bosonic systems.
    In this work, we present a precise definition of quantum phases for non-Hermitian systems and propose a new family of phases referred to as composite quantum phases.
    We demonstrate the existence of these phases in a one-dimensional spin-$1$ system and show their robustness against perturbations through numerical simulations.
    Furthermore, we investigate the phase diagram of our model, indicating the extensive presence of these new phases in non-Hermitian systems.
    Our work establishes a new framework for studying and constructing quantum phases in non-Hermitian interacting systems, revealing exciting possibilities beyond the single-particle picture.
\end{abstract}
\maketitle

\section{Introduction}
Non-Hermitian systems~\cite{Bender2002, Bender2007}, originally proposed as effective theories to describe open systems~\cite{Konotop2016, Vega2017, El2018, Miri2019}, have received significant attention recently due to their unique properties and phenomena beyond the standard Hermitian formalism~\cite{Heiss1990, Peng2015, Shen2018, Kunst2018, Song2019A, Matsumoto2020, Borgnia2020}.
Various studies have focused on the topological properties in free fermion models~\cite{Zeuner2015, Zeng2020}, including the breakdown of the well-known bulk-edge correspondence~\cite{Xiong2018}, inspiring the revisitation of the relationship between bulk topological invariants and edge states in non-Hermitian systems~\cite{Yao2018A, Yao2018B, Song2019B}.
The celebrated Altland-Zirnbauer symmetry classification~\cite{Altland1997} has also been extended from ten to 38 classes by considering additional sublattice symmetries and pseudo-Hermiticity, revealing a much richer phase diagram in non-Hermitian fermionic systems~\cite{Gong2018, Kawabata2019A}.

Topological quantum phases have been extensively investigated in Hermitian interacting bosonic systems besides free fermion models, with a focus on their well-organized entanglement structure.
This long-range entanglement pattern is commonly called topological order~\cite{Levin2006, Kitaev2006, Chen2010}.
In addition, the manifold of symmetric Hamiltonians gives rise to more non-trivial quantum phases, including symmetry-breaking phases~\cite{Landau1950} and symmetry-protected topological (SPT) phases~\cite{Gu2009, Chen2011A, Chen2013}.

The intersection between non-Hermitian physics and many-body physics would be of particular interest.
Previous works have made much progress in this direction, including studies on topological excitations~\cite{Chen2023} or dynamics~\cite{Zhang2020, Turkeshi2023} in interacting spin models~\cite{CastroAlvaredo2009} and observation of non-Hermitian skin effects~\cite{Mu2020, Alsallom2022, Kawabata2023} or many-body localization~\cite{Hamazaki2019, Zhai2020, Wang2023} in interacting fermionic systems.
However, a comprehensive understanding and classification of quantum phases in non-Hermitian interacting systems, i.e., strongly correlated phases at zero temperature, has not been well established.

Here we propose a new definition of quantum phases in non-Hermitian systems by starting from the equivalent classes of Hamiltonians.
With this definition, we demonstrate a broad range of novel non-Hermitian quantum phases without Hermitian counterparts, which we denote as composite quantum phases.
As an illustration, we employ the non-Hermitian parent Hamiltonian method~\cite{Shen2023} to construct a system belonging to this type of new phase and numerically confirm its robustness against perturbations using the multisite infinite time-evolving block decimation (iTEBD) method~\cite{Hastings2009, Shen2023}.
Our results suggest that non-Hermitian composite quantum phases are prevalent, as evidenced by the phase diagram of our model, indicating the existence of a vast and unexplored landscape of non-Hermitian topological phases beyond existing free fermion models.

\section{Quantum phases and quantum phase transitions}
\subsection{Three equivalent criteria in Hermitian systems}
Quantum phases are defined as the equivalent classes of Hamiltonians.
\begin{definition}
    Two local and gapped Hamiltonians in Hermitian systems $H_0$ and $H_1$ belong to the same phase if and only if there exists a set of $H(g)$ connecting them, i.e., $H(0) = H_0$ and $H(1) = H_1$, such that the expectation value of any local observable for the ground state $\braket{O}(g)$ is smooth along the path $g\in[0, 1]$.
\end{definition}
This serves as the original definition for quantum phases and quantum phase transitions.
Thanks to the perturbation theory in Hermitian systems, the existence of an adiabatic path such that $H(g)$ is gapped automatically ensures that $\braket{O}(g)$ is smooth~\cite{Chen2010}, providing another criterion that relates phase transitions to gap closing.

Nevertheless, whether such an adiabatic path exists is generally hard to determine, inspiring people to study the equivalent classes of ground states rather than Hamiltonians.
It is proved by explicit constructions that the above condition is equivalent to the existence of a local unitary (LU) evolution connecting the respective ground states.
This sufficient and necessary condition enables the classification of quantum phases by analyzing ground state properties, such as the entanglement spectrum (ES)~\cite{Li2008, Pollmann2010} or topological entanglement entropy~\cite{Levin2006, Kitaev2006}.
The relation between the classification of Hamiltonians and the classification of states in Hermitian systems is shown in Fig.~\ref{Fig: Scheme}.

\subsection{Non-Hermitian systems}
\begin{figure*}
    \centering
    \includegraphics[width=0.7\linewidth]{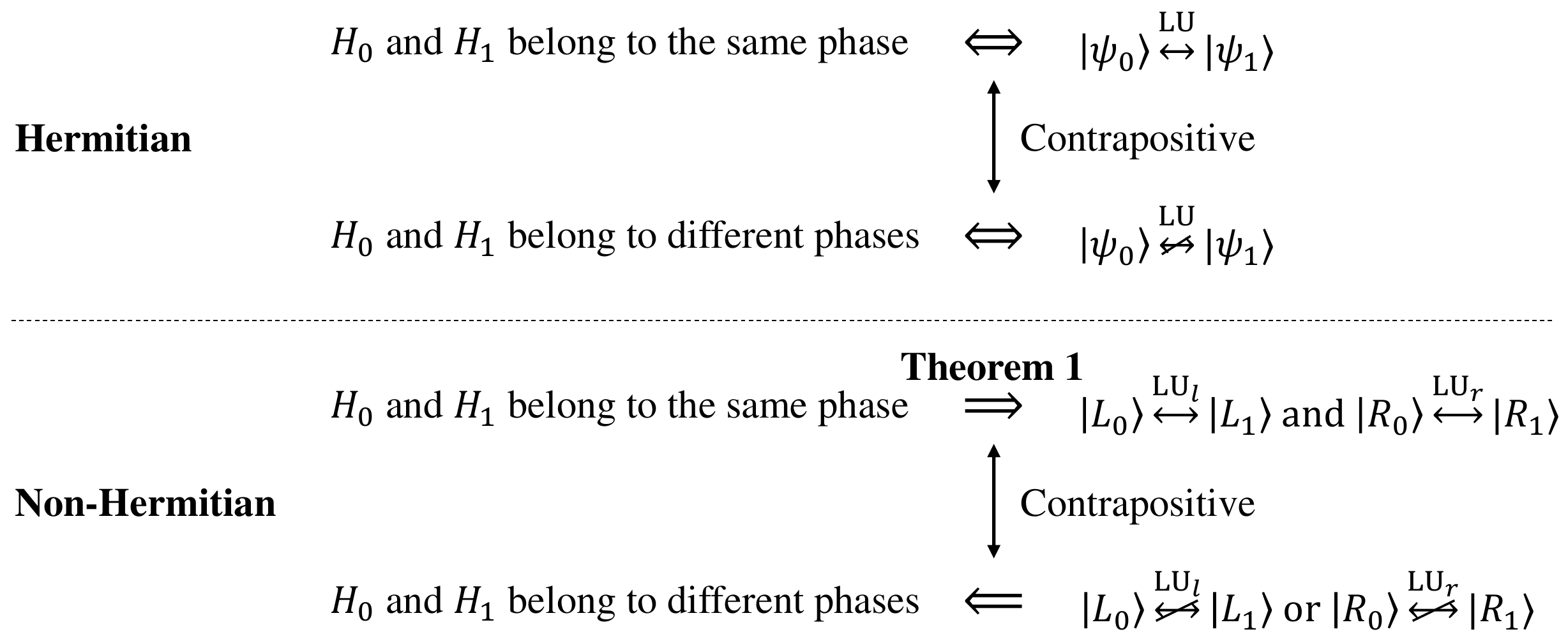}
    \caption{Relations between the classification of Hamiltonians and the classification of quantum states.
    In Hermitian systems, they are equivalent.
    In non-Hermitian systems, Theorem \ref{The: LU} connects the original definition based on smooth observables to the existence of LU evolutions on the corresponding left and right ground states.}
    \label{Fig: Scheme}
\end{figure*}

It is worth noting that in Hermitian systems, 1) the classification of Hamiltonians and 2) the classification of ground states are equivalent, greatly facilitating the exploration and construction of novel quantum phases.
In contrast, in the non-Hermitian regime, the second and third criteria mentioned above and thus the duality between quantum states and Hamiltonians no longer hold.
Therefore, recent studies showing that no new topological ground state can be realized in one-dimensional (1D) interacting non-Hermitian systems~\cite{Xi2021} do not preclude the possibility of investigating new phases of non-Hermitian Hamiltonians, which remains an open problem.

To classify quantum phases in non-Hermitian systems, we adopt the original definition in Hermitian systems, i.e., quantum phases are defined as the equivalent classes of Hamiltonians, where a smooth path for all expectation values is crucial.
However, there are different ways to define the density matrix and evaluate the expectation value of an observable $\braket{O}$ for a general non-Hermitian system~\cite{Brody2013, Lee2020, Grimaudo2020, Grimaldi2021}, each with a self-consistent physical interpretation.
In this work, we make use of the formalism discussed in Refs.~\cite{Brody2013, Herviou2019} based on the bi-orthogonal interpretation of non-Hermitian quantum mechanics, which is to be explained in detail below.

According to the basic formalism in general quantum statistics, the natural choice of the density matrix is $\rho=e^{-\beta H}/\Tr{[e^{-\beta H}]}$.
The expectation value of any observable $\braket{O}$ can then be obtained by
\begin{align}
    \braket{O}
    = \frac{\mathrm{Tr}\left[e^{-\beta H}O\right]}{\mathrm{Tr}\left[e^{-\beta H}\right]}
    = \frac{\sum_{n}{\bra{\psi_n}e^{-\frac{\beta}{2} H}Oe^{-\frac{\beta}{2} H}\ket{\psi_n}}}{\sum_{n}{\bra{\psi_n}e^{-\frac{\beta}{2} H}e^{-\frac{\beta}{2} H}\ket{\psi_n}}}.
\end{align}
As the temperature approaches zero, we reach
\begin{align}
    \braket{O}_{LR} = \frac{\bra{L}O\ket{R}}{\braket{L|R}},
\end{align}
where $\ket{R}$ and $\ket{L}$ are the ground states of $H$ and $H^{\dagger}$ respectively, defined as the eigenstates with the lowest real parts of eigenvalues.
As a result, one can also identify the density matrix for ground states as $\rho = \ket{R}\hspace{-1mm}\bra{L}$ with proper normalization $\braket{L|R} = 1$.
Notably, such a choice for the density matrix has a clear geometric interpretation~\cite{Ju2019} and allows for a natural generalization of the definition for quantum phases from Hermitian systems to the non-Hermitian regime as follows.
\begin{definition}
    Two local, line-gapped~\cite{Kawabata2019B}, non-Hermitian Hamiltonians $H_0$ and $H_1$ belong to the same quantum phase if and only if there exists a set of $H(g)$ connecting them, i.e., $H(0) = H_0$ and $H(1) = H_1$, such that all local observables for the ground states $\Tr{[\rho(g) O]}$, where $\rho(g) = \ket{R(g)}\hspace{-1mm}\bra{L(g)}$ is the density matrix for ground states after normalization $\Tr{[\rho(g)]} = 1$, are smooth along the path $g\in[0, 1]$.
\end{definition}

The first issue we encounter is whether our generalized definition is consistent with the conventional definition of quantum phases for Hermitian Hamiltonians.
From the most basic topological point of view, two regions that are not originally connected may be connected on an extended manifold.
Therefore, we need to answer the following question.
For two Hermitian Hamiltonians $H_0$ and $H_1$ belonging to different quantum phases defined conventionally, can we connect them without phase transitions in the extended manifold of non-Hermitian Hamiltonians?
We will show that if restricted to a special class of systems whose ground states are guaranteed to be short-range correlated and satisfy the entanglement area law, we would not encounter conflicts.
It is noteworthy that it is a condition automatically satisfied in Hermitian systems, but not necessarily in a general non-Hermitian system even with a line gap.

In addition, due to the breakdown of the Lieb-Robinson bound, quantum phase transitions can occur without gap closing in non-Hermitian systems~\cite{Matsumoto2020}.
In other words, a finite gap along the path can no longer guarantee that two Hamiltonians belong to the same phase, hindering us from constructing LU evolution on corresponding ground states and deriving an equivalent criterion as in Hermitian systems.
Therefore, we need to consider how to provide a classification of non-Hermitian quantum phases via another easily implemented criterion.

The following theorem can answer the above two questions.
\begin{theorem}
    For two local, line-gapped, non-Hermitian Hamiltonians $H_0$ and $H_1$ whose ground states are short-range correlated and satisfy the entanglement area law, if they belong to the same quantum phase, their left and right ground states can be connected with LU evolutions respectively, i.e., $\ket{L_0}\hspace{-1mm}\stackrel{\mathrm{LU}_l}{\leftrightarrow}\hspace{-1mm}\ket{L_1}$ and $\ket{R_0}\hspace{-1mm}\stackrel{\mathrm{LU}_r}{\leftrightarrow}\hspace{-1mm}\ket{R_1}$.
    \label{The: LU}
\end{theorem}
\begin{proof}
    Consider the adiabatic path $H(g)$ connecting these two Hamiltonians, i.e., $H(0) = H_0$ and $H(1) = H_1$, where the smoothness of $\Tr{[\ket{R(g)}\hspace{-1mm}\bra{L(g)}O]}$ for any local operator requires all the local reduced density matrices of $\ket{R(g)}\hspace{-1mm}\bra{L(g)}$ be smooth, which further means that those of $\ket{L(g)}$ and $\ket{R(g)}$ are also smooth.
    For each $g$, we can construct the Hermitian parent Hamiltonians $H_L(g)$ and $H_R(g)$ for $\ket{L(g)}$ and $\ket{R(g)}$, respectively, which are local and gapped since both $\ket{L(g)}$ and $\ket{R(g)}$ are short-range correlated and satisfy the entanglement area law~\cite{PerezGarcia2007,PerezGarcia2008}.
    In addition, each term of $H_{L}(g)$ $(H_{R}(g))$, constructed from the local reduced density matrix of $\ket{L(g)}$ $(\ket{R(g)})$, has a smooth dependence on $g$ since the dimension of the local support space does not change.
    Therefore, we obtain two adiabatic paths in the Hermitian Hamiltonian manifold connecting $\ket{L(0)}$ to $\ket{L(1)}$ and $\ket{R(0)}$ to $\ket{R(1)}$, respectively.
    Consequently, $\ket{L(0)}$ and $\ket{L(1)}$ ($\ket{R(0)}$ and $\ket{R(1)}$) can be connected with LU evolutions following the construction in Ref.~\cite{Chen2010}.
\end{proof}
The key step in this proof is to derive an LU evolution for each side by constructing accompanying Hermitian parent Hamiltonians.
From the contraposition of Theorem~\ref{The: LU}, it follows directly that for two Hamiltonians $H_0$ and $H_1$, once one side of their ground states cannot be connected by LU evolutions, they belong to different non-Hermitian phases.
Therefore, the classification of non-Hermitian quantum phases is given by the direct product of the equivalent classes of the left and right ground states.
The new relation in non-Hermitian systems is also shown in Fig.~\ref{Fig: Scheme}.

\section{New phases in non-Hermitian systems}
Following the definition in the previous section, we can propose new quantum phases in non-Hermitian systems.
Intuitively, if a Hermitian system has $n$ different phases, then a non-Hermitian system can potentially exhibit $n\times n$ different phases.
In this case, the $n$ ``diagonal'' phases have Hermitian counterparts, while the ``off-diagonal'' phases are new quantum phases arising only in non-Hermitian systems, where the left and right ground states $\ket{L}$ and $\ket{R}$ cannot be connected by LU evolutions.
Hamiltonians in these new phases, which we denote as composite quantum phases, cannot be adiabatically connected to any conventional Hermitian Hamiltonian.
The schematic diagram of the composite quantum phases is shown in Fig.~\ref{Fig: Composite}.
\begin{figure}
    \centering
    \includegraphics[width=\linewidth]{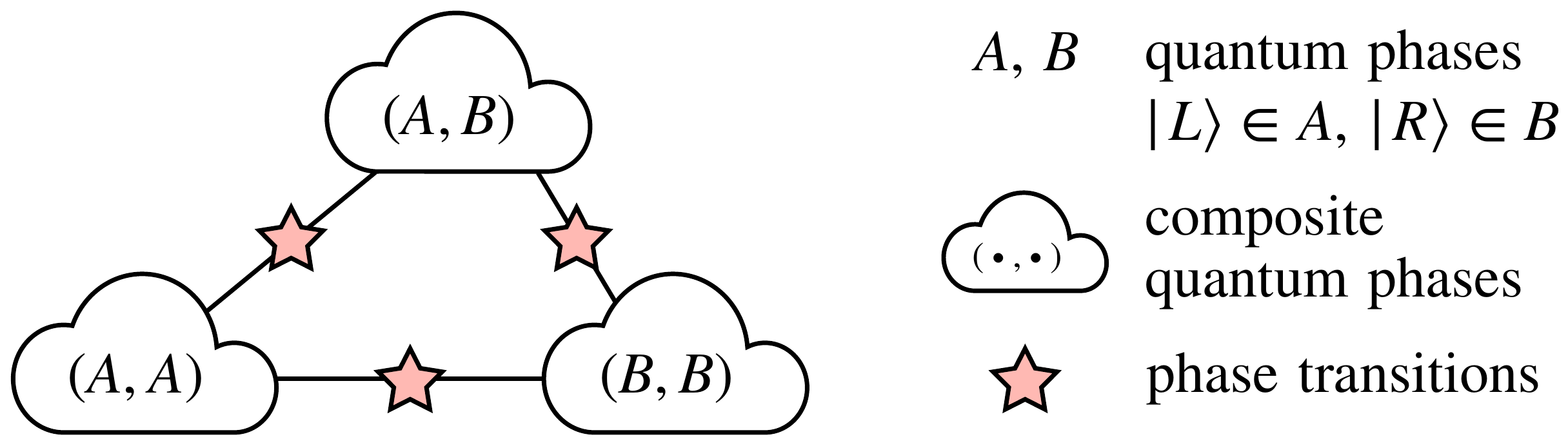}
    \caption{Schematic diagram of composite quantum phases in non-Hermitian systems.
    In this case, the left and right ground states can belong to different phases [e.g., $\ket{L}\in\text{phase }A$ while $\ket{R}\in\text{phase }B$], resulting in composite quantum phases (e.g., the cloud labeled as $(A,B)$).}
    \label{Fig: Composite}
\end{figure}

In addition, when considering symmetric LU evolution, quantum states with different nontrivial SPT orders cannot be transformed into each other~\cite{Gu2009, Chen2011A}.
Therefore, we can extend the above definition to define and study composite symmetry-protected topological (CSPT) orders in non-Hermitian systems.
For instance, in the presence of on-site unitary symmetry defined by a finite group $G$, the classification of SPT phases in $d$-dimensional Hermitian systems can be represented as $\omega \in H^{d+1}(G, \mathbb{C})$~\cite{Chen2011A, Chen2013}.
Consequently, we can use
\begin{align}
    \omega_{L}\times\omega_{R} \in H^{d+1}(G, \mathbb{C})\times H^{d+1}(G, \mathbb{C})
\end{align}
to label possible CSPT phases in non-Hermitian systems where on-site unitary symmetry $G$ is imposed.
In the presence of symmetries besides onsite unitary ones, such as time reversal (TR) or translational invariance (TI), we can further construct additional CSPT phases from states with SPT order that are protected by these joint symmetries~\cite{Liu2011, Chen2011B}.

A significant problem lies in the existence of such composite phases in the real world, i.e., whether we can construct a non-Hermitian parent Hamiltonian from given left and right ground states with different orders such that it remains gapped in the thermodynamic limit.
In Hermitian systems, the existence is guaranteed by the parent Hamiltonian method~\cite{PerezGarcia2007, PerezGarcia2008}.
In the following, we adopt the recently proposed non-Hermitian parent Hamiltonian (nH-PH) approach~\cite{Shen2023} to construct and study composite phases in one dimension.
In this method, one starts from two different matrix product states (MPS)~\cite{PerezGarcia2007} and constructs a local Hamiltonian such that they serve as the zero-energy mode on each side.
Since there is no intrinsic topological order for any injective MPS~\cite{Chen2010}, we will focus on two MPS with different SPT orders.

\section{CSPT with $D_{2h}$ symmetry}
We start from 1D quantum states with different SPT orders protected by the $D_{2h}$ symmetry group, which is a joint symmetry group composed of the dihedral group $D_2$ and the TR transformation $\mathcal{T}$~\cite{Liu2011}.
Different SPT orders are labeled by several indices $\omega, \beta(\mathcal{T}), \gamma(g)$ as defined below.

\subsection{SPT phases with combined sysmmetry}
\begin{figure*}
    \centering
    \includegraphics[width=0.7\linewidth]{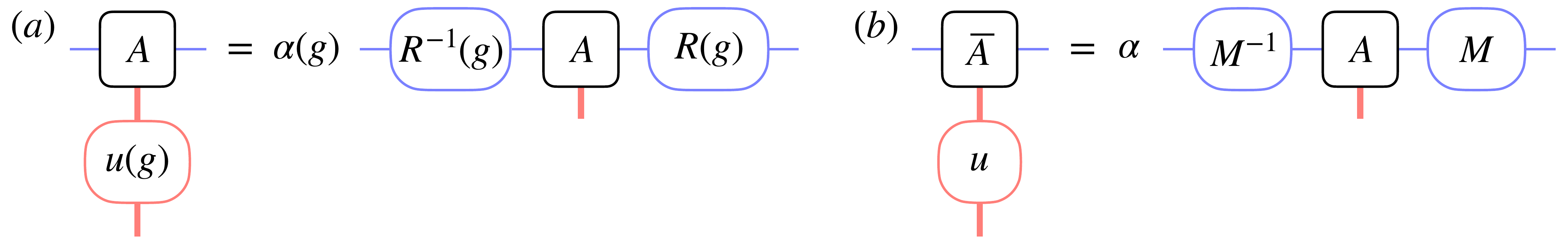}
    \caption{Transformation of the local tensor $A$ of a symmetric MPS under (a) onsite symmetry $g\in G$ and (b) time reversal symmetry $\mathcal{T} = uK$.}
    \label{Fig: S2}
\end{figure*}
Here we briefly review the definition and indices for different SPT orders in 1D Hermitian systems~\cite{Liu2011, Chen2011A, Chen2011B, Pollmann2012}.
Without loss of generality, we consider an MPS with TI, but we do not always impose TI for the LU evolution.
In this case, onsite unitary symmetry must act linearly, not projectively~\cite{Chen2011A}.
For an onsite unitary symmetry group $G$, we have
\begin{align}
    \sum_{j}{u(g)_{ij}A^{[j]}} = \alpha(g)R^{-1}(g)A^{[i]}R(g),\quad g\in G
\end{align}
as shown in Fig.~\ref{Fig: S2}(a), where $u(g)$, $\alpha(g)$, and $R(g)$ are a linear representation, a 1D representation, and a projective representation of $G$, respectively.
If only onsite symmetries are imposed, different SPT phases are labeled by $\omega\in H^2\left(G, \mathbb{C}\right)$, defined as
\begin{align}
R(g_z)^{-1} R(g_x)R(g_z) &= \omega R(g_x).
\end{align}

Then we turn to the onsite antiunitary symmetry, e.g., the TR symmetry $\mathcal{T} = uK$ satisfying $\mathcal{T}^2=u\overline{u} = \pm 1$, where $K$ refers to the complex conjugate.
The local tensor is now transformed as
\begin{align}
    \sum_{j}{u_{ij}\overline{A^{[j]}}} = \alpha M^{-1}A^{[i]}M,\label{equ: TR}
\end{align}
as shown in Fig.~\ref{Fig: S2}(b), where $M\overline{M} = \pm 1$. 
Therefore, different quantum phases can be labeled by $M\overline{M} = \pm 1\equiv \beta{(\mathcal{T})}$.

Now we consider the combination of onsite unitary symmetry group $G$ and time reversal symmetry $\mathcal{T}$.
In addition to the indices given by these two components respectively, i.e., $\beta(\mathcal{T}) = \pm 1$ and projective representation $\omega\in H^2(G, \mathbb{C})$, we need to consider the `projective' commutation relation between $G$ and $\mathcal{T}$.
Suppose all elements in $G$ commute with $\mathcal{T}$ (and thus the joint symmetry group is Abelian), one can derive the following relation
\begin{align}
    M^{-1} R(g)M &= \gamma(g)\overline{R(g)},\label{equ: G+TR}
\end{align}
where $\gamma(g)$ is a 1D representation of $G$.
Moreover, two $\gamma(g)$ are equivalent if they differ only by the square of a 1D representation of $G$.
This means that different SPT phases are further distinguished by $\gamma(g)\in \mathcal{G}/\mathcal{G}_2$, where $\mathcal{G}$ is the group of 1D representation of $G$, and $\mathcal{G}_2$ is the group of 1D representation squared of $G$~\cite{Chen2011B}.

\subsection{SPT phases with $D_{2h}$ symmetry in spin-$1$ systems}
As an example, we consider a spin-$1$ chain with $D_2$ symmetry, whose group elements are $\{e, g_x, g_z, g_y = g_xg_z = g_zg_x\}$.
The linear representation of this group applied to the physical bond is $u(g) = e^{-\mathrm{i}\pi s(g)}$, where $s(g_i) \equiv s_i$ for $i=x, y, z$, referring to conventional spin operators.
There is a nontrivial projective representation of $D_2$, labeled as $\omega = -1$, corresponding to the well-known Haldane phase.
A realization for this phase is the AKLT state, which can be represented by an MPS with bond dimension $D = 2$.
The local tensors are $A^{[x]} = X$, $A^{[y]} = Y$, and $A^{[z]} = Z$, where $X, Y, Z$ are the Pauli matrices~\cite{Zeng2019}.

As for the joint group of $D_2\times\mathcal{T}=D_{2h}$, one can realize four distinct SPT phases in spin-$1$ systems, whose MPS constructions with $D=2$ and all associated indices are shown as follows (all these four states have $R(g_x) = X$, $R(g_z) = Z$, and thus $\omega=-1$)~\cite{Liu2011}
\begin{align*}
    \begin{array}{ccccccccc}\hline\hline
    &A^{[x]} &A^{[y]} &A^{[z]} &M &\omega &\beta &\gamma(g_x) &\gamma(g_z)\\\hline
    \ket{\psi_0} &X &Y &Z &Y &-1 &-1 &-1 &-1\\
    \ket{\psi_x} &\mathrm{i}X &Y &Z &Z &-1 &+1 &-1 &+1\\
    \ket{\psi_y} &X &\mathrm{i}Y &Z &I &-1 &+1 &+1 &+1\\
    \ket{\psi_z} &X &Y &\mathrm{i}Z &X &-1 &+1 &+1 &-1\\\hline\hline
    \end{array}\label{equ: four-states}
\end{align*}
where we adopt $\mathcal{T} = e^{-\mathrm{i}\pi s_y}K$, following the conventional choice for spin systems.
It should be noted that these four states all belong to the Haldane phase if $D_2$ symmetry is the only constraint, where $\omega=\pm 1$ enables us to distinguish the Haldane phase from the trivial phase, while they can be further distinguished by different commutation relations between $g\in D_2$ and $\mathcal{T}$ [labeled $\gamma(g)$].

\subsection{CSPT phases with $D_{2h}$ symmetry in spin-$1$ systems}
We proceed to establish CSPT phases in non-Hermitian systems based on these states in the following.
We construct a Hermitian parent Hamiltonian $H_{00}$ to describe the ground state $\ket{\psi_0}\hspace{-1mm}\bra{\psi_0}$ and a non-Hermitian parent Hamiltonian $H_{x0}$ from $\ket{\psi_0}\hspace{-1mm}\bra{\psi_x}$~\cite{Shen2023}.
We choose $k=4$ in each construction, i.e., both $H_{00}$ and $H_{x0}$ involve four-site interaction.
Furthermore, both Hamiltonians preserve the $D_{2h}$ symmetry inherited from $\ket{\psi_0}$ and $\bra{\psi_x}$.
Thus, $H_{00}$ and $H_{x0}$ are expected to belong to different non-Hermitian phases.
To investigate the quantum phase transition between $H_{00}$ and $H_{x0}$, we consider the path of Hamiltonians given by $H_0(\lambda) = (1-\lambda)H_{00}+\lambda H_{x0}$ for $\lambda\in[0, 1]$.
We note that $\ket{\psi_0}$ is always a zero-energy eigenstate of $H_0(\lambda)$ for all $\lambda$ since it is a co-eigenstate of both $H_{00}$ and $H_{x0}$ with energy $E = 0$.

To compute the left and right ground states $\bra{L}$ and $\ket{R}$ of $H_0(\lambda)$, we employ the multisite iTEBD method~\cite{Shen2023} with bond dimension $D = 32$ and unit cell length $k=4$.
We set the time step as $\Delta\tau=1\times 10^{-2}$, and adopt the convergence criterion $e=1\times 10^{-12}$, defined as $e=\sum_{i=1}^{k}\sum_{j=1}^{D} \left[s_{ij}(\tau + \delta\tau) - s_{ij}(\tau)\right]^2$, where $s_{ij}$ denotes the $j$-th Schmidt weight for site $i$ in the unit cell.
The resulting entanglement spectra (ES) of $\ket{L}$ and $\ket{R}$ are shown in Figs.~\ref{Fig: iTEBD}(a) and \ref{Fig: iTEBD}(b).
\begin{figure}
    \centering
    \includegraphics[width=\linewidth]{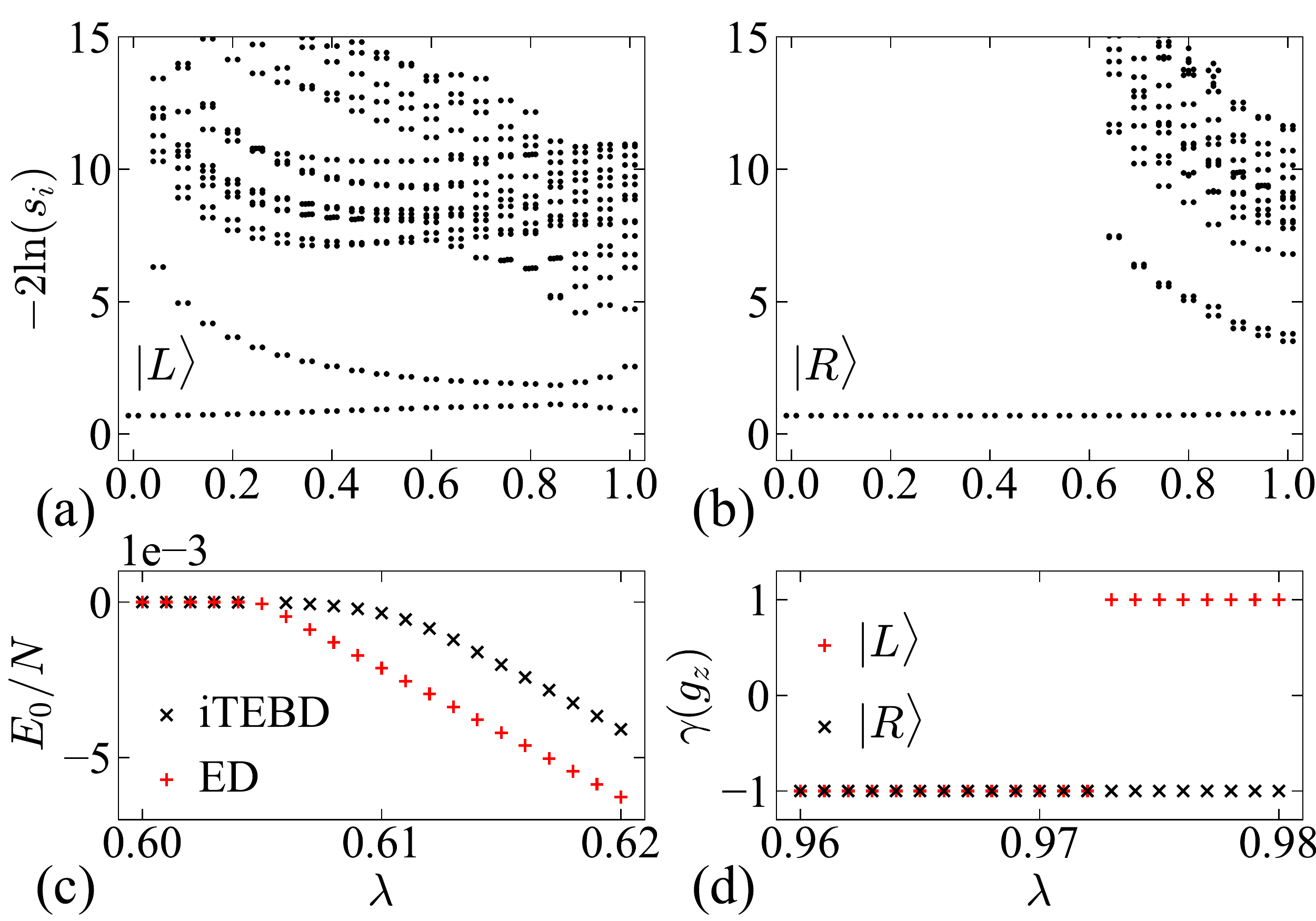}
    \caption{iTEBD calculation for $\ket{L}$ and $\ket{R}$ of $H_0(\lambda)$ with $D = 32$.
    (a and b) Entanglement spectrum of $\ket{L}$ and $\ket{R}$ respectively.
    (c) The ground state energy per site $E_0/N$ of $H_0(\lambda)$, compared with the results calculated from ED with $N = 10$ and PBC.
    (d) Index $\gamma(g_z)$ to describe the SPT order protected by $D_{2h}$ of $\ket{L}$ and $\ket{R}$.}
    \label{Fig: iTEBD}
\end{figure}

Firstly, it is observed that the ES of $\ket{R}$ undergoes an abrupt change when the parameter $\lambda$ approaches the critical value $\sim 0.6$.
For $\lambda\leq 0.6$, the ES of $\ket{R}$ coincides with that of $\ket{\psi_0}$. 
Consequently, $\ket{\psi_0}$ is deemed to be a zero mode as well as the true ground state of the Hamiltonian $H$ with associated energy $E_0 = 0$.
However, when $\lambda$ surpasses $0.6$, the ES of $\ket{R}$ retains its twofold degeneracy while exhibiting additional smaller Schmidt weights.
This result indicates that $\ket{R}$ also possesses non-trivial SPT order~\cite{Pollmann2010} despite the fact that $\ket{\psi_0}$ is no longer the ground state.
Combined with the ground state energy per site $E_0/N$ calculated by exact diagonalization (ED) for finite systems with $N = 10$ under periodic boundary conditions (PBC) and by iTEBD for infinite systems, which are real in the entire region as illustrated in Fig.~\ref{Fig: iTEBD}(c), we observe a first-order phase transition occurring at $\lambda_{c1} \approx 0.608$, where the ground state and the first excited state undergo a level crossing.
The origin of this phase transition can be attributed to the breakdown of the variational principle in non-Hermitian systems, whereby the summation of ground state energy of each term in the Hamiltonian (which equals zero in this scenario) cannot be utilized to provide a lower bound for the spectrum of the entire Hamiltonian.

Figures~\ref{Fig: iTEBD}(a) and \ref{Fig: iTEBD}(b) demonstrate that both $\ket{L}$ and $\ket{R}$ exhibit nontrivial SPT order throughout the path $\lambda\in[0, 1]$.
To identify the specific SPT phases to which they belong, and to further classify the non-Hermitian phase diagram of $H_0(\lambda)$, we calculate the index $\gamma(g_z)$ for $\ket{L}$ and $\ket{R}$ along the path.
It should be noted that the index differs for $\ket{\psi_0}$ [$\gamma(g_z) = -1$] and $\ket{\psi_x}$ [$\gamma(g_z) = 1$].
Details on the calculation method for this index can be found in the Appendix, while the results are presented in Fig.~\ref{Fig: iTEBD}(d).

Our analysis shows that $\gamma(g_z)=-1$ for $\ket{R}$ throughout the path $\lambda\in[0, 1]$, which is the same as that of $\ket{\psi_0}$.
On the other hand, $\gamma(g_z)$ of $\ket{L}$ transitions from $-1$ to $+1$ (the same as that of $\ket{\psi_x}$) at $\lambda_{c2}\approx 0.973$, describing a phase transition of $\ket{L}$ between different SPT phases.
As a result, $H_0(\lambda)$ for $\lambda>\lambda_{c2}$ exhibits a nontrivial CSPT order, where $\ket{L}$ and $\ket{R}$ belong to different SPT phases protected by $D_{2h}$.
In other words, the topological property of the ground state undergoes a qualitative change at $\lambda_{c2}$, suggesting a phase transition between the conventional Haldane phase and the newly discovered CSPT phase without gap closing.
This phase transition originates from the breakdown of the well-known Lieb-Robinson bound in non-Hermitian systems~\cite{Matsumoto2020}, where the Hermitian perturbation theory is not directly applicable.

\subsection{CSPT under perturbation and the phase diagram.}
The robustness of the CSPT phase constructed in the previous section against perturbations that preserve $D_{2h}$ symmetry is demonstrated in this section.
Specifically, we introduce an on-site potential energy term $U{S^z}^2$ to the Hamiltonian previously considered, i.e., $H(\lambda, U) = H_0(\lambda) + \sum_{i}U{S_i^z}^2$.
We begin with the case where $U\ll 1$, and we fix $\lambda=1$ while gradually increasing $U$.
Our numerical simulations indicate that the constructed CSPT phase, identified by different $\gamma(g_z)$ values for $\ket{L}$ and $\ket{R}$, is robust with the appearance of $U$ up to $0.1$.
This is shown in the rightmost column of the phase diagram depicted in Fig.~\ref{Fig: PhaseDiagram}(a), which will be discussed later.

\begin{figure}
    \centering
    \includegraphics[width=\linewidth]{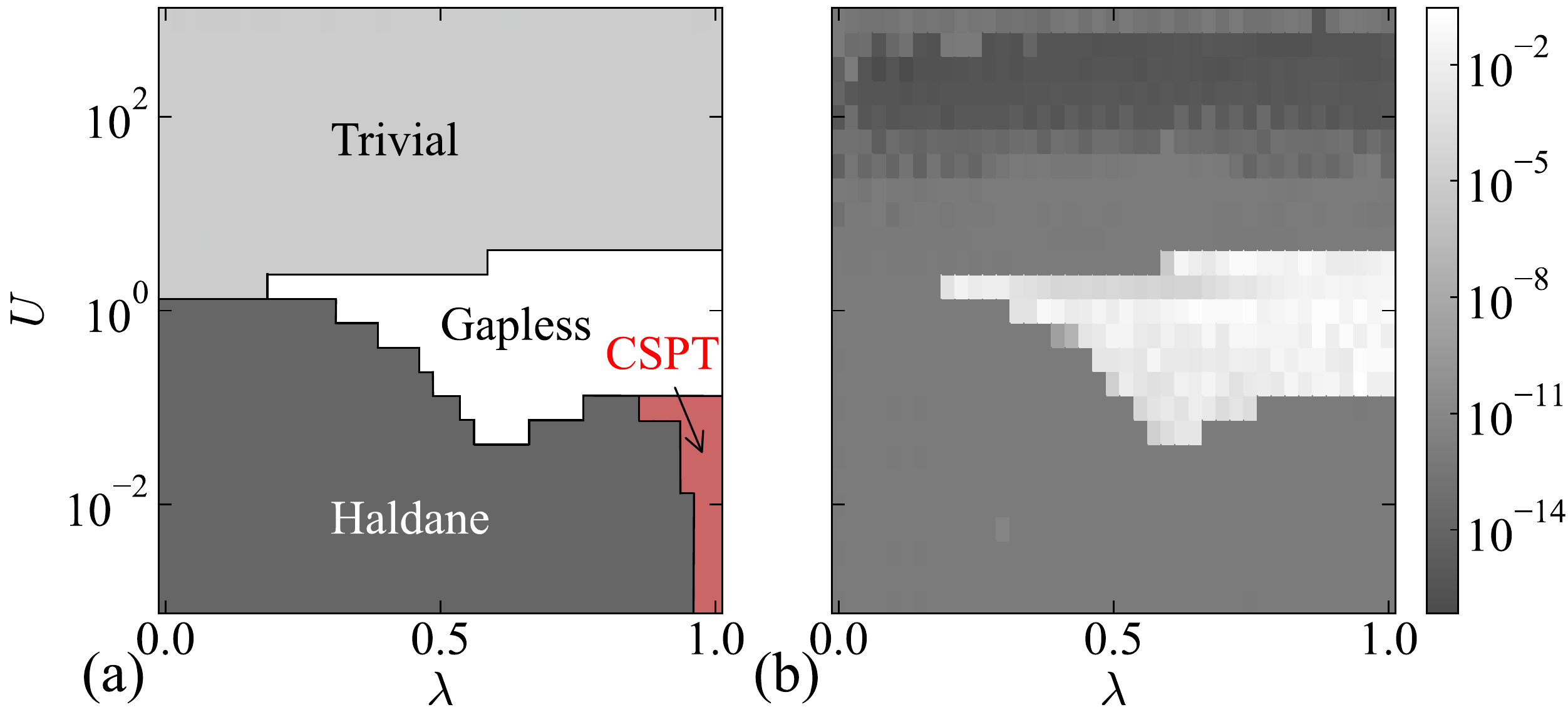}
    \caption{Phase diagram of $H(\lambda, U)$.
    (a) $RGB$ values assigned as $\left[\gamma(g_z)_{\ket{L}}, \gamma(g_z)_{\ket{R}}, \omega\right]/5+0.6$, respectively.
    (b) Residual error after $50000$ steps of iteration in the iTEBD method.}
    \label{Fig: PhaseDiagram}
\end{figure}

Another extreme limit is the onsite potential term dominating the Hamiltonian, i.e.,  $U\gg 1$, where both $\ket{L}$ and $\ket{R}$ can be adiabatically connected to product states.
This indicates that the Hamiltonian $H(\lambda, U)$ belongs to a trivial symmetric phase for large values of $U$.
To visually represent these three quantum phases in a phase diagram, we also consider the conventional index $\omega$ that distinguishes the Haldane phase from the trivial symmetric phase.
The trivial symmetric state corresponds to $\omega=1$, while all four SPT phases protected by $D_{2h}$ have $\omega=-1$.
Therefore, we can use $\omega$, $\gamma(g_z)_{\ket{L}}$, $\gamma(g_z)_{\ket{R}}$ as the joint indicator to identify the three phases as listed below.
\begin{align*}
    \begin{array}{ccccccc}\hline\hline
    \multirow{2}{*}{Phase of $H(\lambda, U)$} & \multicolumn{3}{c}{\textrm{Index of }\ket{L}}& \multicolumn{3}{c}{\textrm{Index of }\ket{R}}\\
    ~&\omega &\gamma(g_z) &\gamma(g_x) &\omega &\gamma(g_z) &\gamma(g_x)\\\hline
    \text{Trivial} & +1 & +1 & +1 & +1 & +1 & +1\\
    \text{Haldane} & -1 & -1 & -1 & -1 & -1 & -1\\
    \text{CSPT} & -1 & +1 & -1 & -1 & -1 & -1\\\hline\hline
    \end{array}
\end{align*}

In Fig.~\ref{Fig: PhaseDiagram}(a), we use three channels of $RGB$ to demonstrate three indices to be considered, where the values of indices are rescaled from $[-1, +1]$ to $[0.4, 0.8]$ as $RGB$ values, i.e., $[R, G, B] = \left[\gamma(g_z)_{\ket{L}}, \gamma(g_z)_{\ket{R}}, \omega\right]/5+0.6$.
Meanwhile, the residual error $e$ of iTEBD after $50000$ steps of iteration is shown in Fig.~\ref{Fig: PhaseDiagram}(b).
Only gapped quantum phases with convergent ES were considered, while the nonconvergent region was attributed to the gapless nature of the systems.
The phase diagram clearly demonstrates the existence of three gapped quantum phases of $H(\lambda, U)$ discussed above.
Notably, the range of $\lambda$ for the existence of CSPT is observed to increase with enhancing $U$, indicating that the newly-established composite quantum phases can exist extensively in non-Hermitian systems.
Another noteworthy observation in the phase diagram is the absence of a direct phase transition between the trivial symmetric phase and the CSPT phase.
Further investigation is required to determine whether such a phase transition can exist and the underlying reasons.

\section{Conclusions and discussions}
In this work, we clarify the definition of quantum phases and quantum phase transitions in non-Hermitian systems.
Specifically, we prove that if two local, line gapped, non-Hermitian Hamiltonians belong to the same quantum phase, their left and right ground states can be adiabatically connected respectively.
This holds true provided that the ground state manifold is short-range correlated and satisfies the entanglement area law.

Based on this definition, we propose a novel class of quantum phases in non-Hermitian systems, denoted as composite quantum phases, whose left and right ground states belong to different phases.
Furthermore, this definition can be extended to define the CSPT order subject to an additional symmetry restriction.

The recently proposed parent Hamiltonian method for non-Hermitian systems has enabled us to construct a system that can realize CSPT phases protected by the $D_{2h}$ symmetry group. 
Through numerical verification using the iTEBD algorithm, we demonstrate the existence of this type of new phase and investigate the phase diagram after introducing an on-site potential term that preserves symmetry.
Our results show that the CSPT phase is not only robust against symmetric perturbations but also has a substantial region of existence in our phase diagram.

This study provides a new perspective for the systematic understanding, classification, and construction of novel quantum phases in non-Hermitian systems.
Moreover, these composite quantum phases lack Hermitian counterparts, suggesting a vast field for exploration in non-Hermitian many-body physics.

\section{Acknowledgments}
This work is supported by the National Natural Science Foundation of China (NSFC) (Grants No. 12174214 and No. 92065205), the National Key R\&D Program of China (Grant No. 2018YFA0306504), the Innovation Program for Quantum Science and Technology (Grant No. 2021ZD0302100), and the Tsinghua University Initiative Scientific Research Program.

\section*{Appendix}
\renewcommand{\theequation}{S\arabic{equation}} \setcounter{equation}{0}
\renewcommand{\thefigure}{S\arabic{figure}} \setcounter{figure}{0}

\subsection{Density matrix and Expectation values for ground states}
In this work, the expectation value for a general non-Hermitian system at zero temperature is calculated as
\begin{align}
    \braket{O}_{LR} = \frac{\bra{L}O\ket{R}}{\braket{L|R}},
\end{align}
where $\ket{R}$ and $\ket{L}$ are the ground states of $H$ and $H^{\dagger}$, respectively, defined as the eigenstates with the lowest real parts of eigenvalues.
It also leads to the density matrix for ground states as $\rho = \ket{R}\hspace{-1mm}\bra{L}$ with proper normalization $\braket{L|R} = 1$.

The expectation value $\braket{O}_{LR}$ in the above formalism may generally be complex.
However, some previous studies adopted a similar formalism to calculate expectation values, where all their considered expectations were real and thus physically detectable and meaningful, including chiral order and string order~\cite{Shen2023} or the particle density~\cite{Lee2020}.
\citet{Chen2023} proposed a method to calculate the general spectral functions for non-Hermitian systems defined in the same formalism.
At the same time, the physical interpretation of $\braket{O}_{LR}$ for a generic operator $O$ was also been addressed~\cite{Brody2013}.
In addition, probing complex energy was reported in ion trap systems recently~\cite{Cao2023}, making the detection of other complex expectations also possible in experiments and providing them with real physical meanings.

\subsection{Non-Hermitian parent Hamiltonian}
In this section, we briefly introduce the non-Hermitian parent Hamiltonian (nH-PH) method proposed by~\citet{Shen2023}.
In this method, one starts from two MPS $\bra{L}$ and $\ket{R}$ with translational invariance (TI), given by
\begin{align}
    \ket{R(L)} = \sum\limits_{i_1,\dots,i_N} \mathrm{Tr}\left[A_{R(L)}^{[i_1]}\dots A_{R(L)}^{[i_N]}\right] \ket{i_1,\dots,i_N}.
\end{align}
After grouping $k$ adjacent tensors as shown in Figs.~\ref{Fig: S1}(a) and \ref{Fig: S1}(b), we need to construct a local projector onto the local support space of each side, i.e., 
\begin{align}
    \hat{P}\hat{T}_R = \hat{T}_R,\qquad \hat{T}_L^{\dagger}\hat{P} = \hat{T}_L^{\dagger}.\label{ProjectorCondition}
\end{align}
It is proved that the only solution is
\begin{align}
    \hat{P} = \hat{T}_R\hat{C}\hat{T}_L^{\dagger},\quad \hat{C} = \hat{G}^{-1},
\end{align}
where $\hat{G}\equiv \hat{T}_L^{\dagger}\hat{T}_R$ is the metric operator shown in Fig.~\ref{Fig: S1}(c).
Therefore, once $\hat{G}$ is invertible, one can define the $k$-local Hamiltonian $\hat{H}= \sum_i\hat{\Pi}_i = \sum_i \left(\hat{I}-\hat{P}_i\right)$ shown in Fig.~\ref{Fig: S1}(d) such that $\bra{L}$ and $\ket{R}$ serve as the zero-energy modes for each local term, and thus for the entire Hamiltonian.
\begin{figure}[H]
    \centering
    \includegraphics[width=\linewidth]{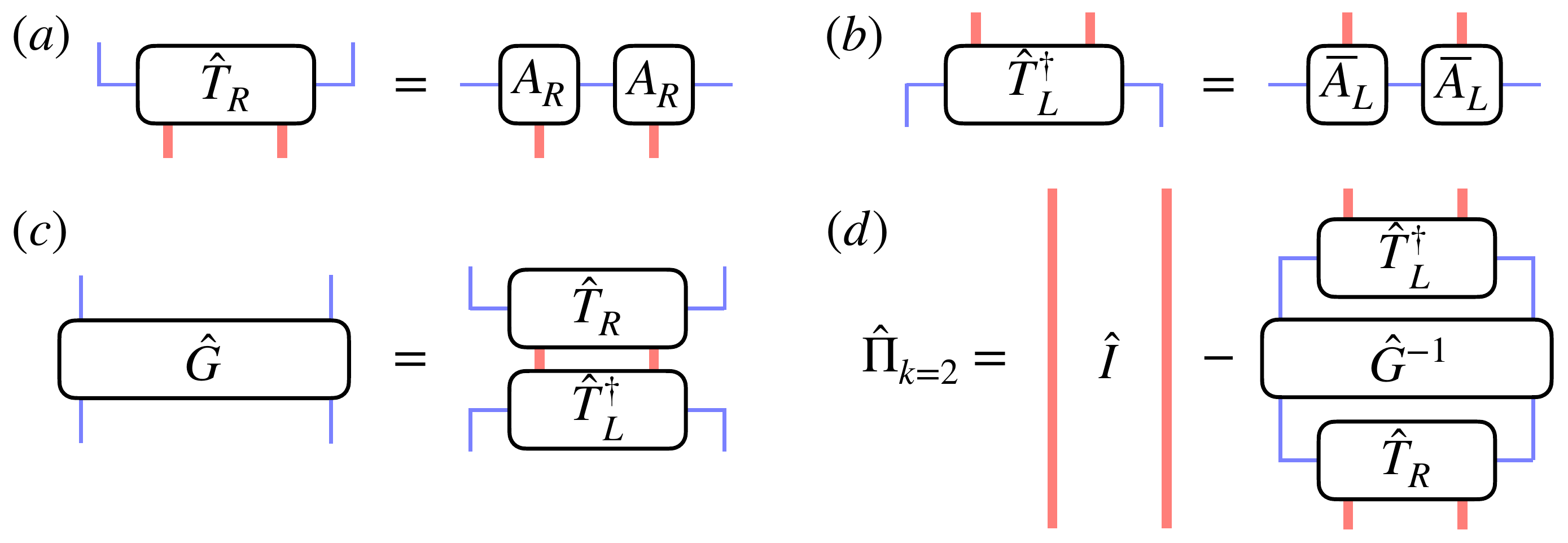}
    \caption{Schematic diagram of non-Hermitian parent Hamiltonian for $k=2$.
    (a and b) Local tensors $\hat{T}_R$ and $\hat{T}_L^{\dagger}$.
    (c) The metric operator $\hat{G} = \hat{T}_L^{\dagger}\hat{T}_R$.
    (d) The local projector $\hat{\Pi} = \hat{I} - \hat{T}_R\hat{C}\hat{T}_L^{\dagger} = \hat{I} - \hat{T}_R\hat{G}^{-1}\hat{T}_L^{\dagger}$.}
    \label{Fig: S1}
\end{figure}

\subsection{Calculation of the indices for SPT phases with combined symmetry}
Formally, the indices to identify different SPT phases can be evaluated as
\begin{align}
    \omega &= R(g_x)^{-1}R(g_z)^{-1}R(g_x)R(g_z)\label{equ: omega}\\
    \beta &= M\overline{M}\label{equ: beta}\\
    \gamma(g) &= \overline{R(g)}^{-1}M^{-1}R(g)M, \quad g\in D_2,\label{equ: gamma}
\end{align}
where the calculation of $\alpha(g)$, $R(g)$, and $M$ is fully discussed by \citet{Pollmann2012}.
However, there are still random global phases in the definition of $M$ and each $R(g)$, which do not affect $\omega$ and $\beta$ since they will be eliminated in Eqs.~\eqref{equ: omega} and \eqref{equ: beta} but must be taken into consideration when calculating $\gamma$ as Eq.~\eqref{equ: gamma} involves a complex conjugate.
To determine this phase factor, we choose the same phase structure as that of the AKLT state, i.e., $R(g)^2 = I$ for $g\in D_2$, which is consistent with all four states $\ket{\psi_0}$, $\ket{\psi_x}$, $\ket{\psi_y}$, $\ket{\psi_z}$ discussed in the main text.
\bibliography{ref}

\begin{thebibliography}{58}%
\makeatletter
\providecommand \@ifxundefined [1]{%
 \@ifx{#1\undefined}
}%
\providecommand \@ifnum [1]{%
 \ifnum #1\expandafter \@firstoftwo
 \else \expandafter \@secondoftwo
 \fi
}%
\providecommand \@ifx [1]{%
 \ifx #1\expandafter \@firstoftwo
 \else \expandafter \@secondoftwo
 \fi
}%
\providecommand \natexlab [1]{#1}%
\providecommand \enquote  [1]{``#1''}%
\providecommand \bibnamefont  [1]{#1}%
\providecommand \bibfnamefont [1]{#1}%
\providecommand \citenamefont [1]{#1}%
\providecommand \href@noop [0]{\@secondoftwo}%
\providecommand \href [0]{\begingroup \@sanitize@url \@href}%
\providecommand \@href[1]{\@@startlink{#1}\@@href}%
\providecommand \@@href[1]{\endgroup#1\@@endlink}%
\providecommand \@sanitize@url [0]{\catcode `\\12\catcode `\$12\catcode
  `\&12\catcode `\#12\catcode `\^12\catcode `\_12\catcode `\%12\relax}%
\providecommand \@@startlink[1]{}%
\providecommand \@@endlink[0]{}%
\providecommand \url  [0]{\begingroup\@sanitize@url \@url }%
\providecommand \@url [1]{\endgroup\@href {#1}{\urlprefix }}%
\providecommand \urlprefix  [0]{URL }%
\providecommand \Eprint [0]{\href }%
\providecommand \doibase [0]{https://doi.org/}%
\providecommand \selectlanguage [0]{\@gobble}%
\providecommand \bibinfo  [0]{\@secondoftwo}%
\providecommand \bibfield  [0]{\@secondoftwo}%
\providecommand \translation [1]{[#1]}%
\providecommand \BibitemOpen [0]{}%
\providecommand \bibitemStop [0]{}%
\providecommand \bibitemNoStop [0]{.\EOS\space}%
\providecommand \EOS [0]{\spacefactor3000\relax}%
\providecommand \BibitemShut  [1]{\csname bibitem#1\endcsname}%
\let\auto@bib@innerbib\@empty
\bibitem [{\citenamefont {Bender}\ \emph {et~al.}(2002)\citenamefont {Bender},
  \citenamefont {Brody},\ and\ \citenamefont {Jones}}]{Bender2002}%
  \BibitemOpen
  \bibfield  {author} {\bibinfo {author} {\bibfnamefont {C.~M.}\ \bibnamefont
  {Bender}}, \bibinfo {author} {\bibfnamefont {D.~C.}\ \bibnamefont {Brody}},\
  and\ \bibinfo {author} {\bibfnamefont {H.~F.}\ \bibnamefont {Jones}},\
  }\bibfield  {title} {\bibinfo {title} {Complex extension of quantum
  mechanics},\ }\href {https://doi.org/10.1103/PhysRevLett.89.270401}
  {\bibfield  {journal} {\bibinfo  {journal} {Phys. Rev. Lett.}\ }\textbf
  {\bibinfo {volume} {89}},\ \bibinfo {pages} {270401} (\bibinfo {year}
  {2002})}\BibitemShut {NoStop}%
\bibitem [{\citenamefont {Bender}(2007)}]{Bender2007}%
  \BibitemOpen
  \bibfield  {author} {\bibinfo {author} {\bibfnamefont {C.~M.}\ \bibnamefont
  {Bender}},\ }\bibfield  {title} {\bibinfo {title} {Making sense of
  non-hermitian hamiltonians},\ }\href
  {https://doi.org/10.1088/0034-4885/70/6/R03} {\bibfield  {journal} {\bibinfo
  {journal} {Rep. Prog. Phys.}\ }\textbf {\bibinfo {volume} {70}},\ \bibinfo
  {pages} {947} (\bibinfo {year} {2007})}\BibitemShut {NoStop}%
\bibitem [{\citenamefont {Konotop}\ \emph {et~al.}(2016)\citenamefont
  {Konotop}, \citenamefont {Yang},\ and\ \citenamefont
  {Zezyulin}}]{Konotop2016}%
  \BibitemOpen
  \bibfield  {author} {\bibinfo {author} {\bibfnamefont {V.~V.}\ \bibnamefont
  {Konotop}}, \bibinfo {author} {\bibfnamefont {J.}~\bibnamefont {Yang}},\ and\
  \bibinfo {author} {\bibfnamefont {D.~A.}\ \bibnamefont {Zezyulin}},\
  }\bibfield  {title} {\bibinfo {title} {Nonlinear waves in
  $\mathcal{PT}$-symmetric systems},\ }\href
  {https://doi.org/10.1103/RevModPhys.88.035002} {\bibfield  {journal}
  {\bibinfo  {journal} {Rev. Mod. Phys.}\ }\textbf {\bibinfo {volume} {88}},\
  \bibinfo {pages} {035002} (\bibinfo {year} {2016})}\BibitemShut {NoStop}%
\bibitem [{\citenamefont {de~Vega}\ and\ \citenamefont
  {Alonso}(2017)}]{Vega2017}%
  \BibitemOpen
  \bibfield  {author} {\bibinfo {author} {\bibfnamefont {I.}~\bibnamefont
  {de~Vega}}\ and\ \bibinfo {author} {\bibfnamefont {D.}~\bibnamefont
  {Alonso}},\ }\bibfield  {title} {\bibinfo {title} {Dynamics of non-markovian
  open quantum systems},\ }\href {https://doi.org/10.1103/RevModPhys.89.015001}
  {\bibfield  {journal} {\bibinfo  {journal} {Rev. Mod. Phys.}\ }\textbf
  {\bibinfo {volume} {89}},\ \bibinfo {pages} {015001} (\bibinfo {year}
  {2017})}\BibitemShut {NoStop}%
\bibitem [{\citenamefont {El-Ganainy}\ \emph {et~al.}(2018)\citenamefont
  {El-Ganainy}, \citenamefont {Makris}, \citenamefont {Khajavikhan},
  \citenamefont {Musslimani}, \citenamefont {Rotter},\ and\ \citenamefont
  {Christodoulides}}]{El2018}%
  \BibitemOpen
  \bibfield  {author} {\bibinfo {author} {\bibfnamefont {R.}~\bibnamefont
  {El-Ganainy}}, \bibinfo {author} {\bibfnamefont {K.~G.}\ \bibnamefont
  {Makris}}, \bibinfo {author} {\bibfnamefont {M.}~\bibnamefont {Khajavikhan}},
  \bibinfo {author} {\bibfnamefont {Z.~H.}\ \bibnamefont {Musslimani}},
  \bibinfo {author} {\bibfnamefont {S.}~\bibnamefont {Rotter}},\ and\ \bibinfo
  {author} {\bibfnamefont {D.~N.}\ \bibnamefont {Christodoulides}},\ }\bibfield
   {title} {\bibinfo {title} {Non-hermitian physics and {PT} symmetry},\ }\href
  {https://www.nature.com/articles/nphys4323} {\bibfield  {journal} {\bibinfo
  {journal} {Nat. Phys.}\ }\textbf {\bibinfo {volume} {14}},\ \bibinfo {pages}
  {11} (\bibinfo {year} {2018})}\BibitemShut {NoStop}%
\bibitem [{\citenamefont {Miri}\ and\ \citenamefont {Alù}(2019)}]{Miri2019}%
  \BibitemOpen
  \bibfield  {author} {\bibinfo {author} {\bibfnamefont {M.-A.}\ \bibnamefont
  {Miri}}\ and\ \bibinfo {author} {\bibfnamefont {A.}~\bibnamefont {Alù}},\
  }\bibfield  {title} {\bibinfo {title} {Exceptional points in optics and
  photonics},\ }\href {https://doi.org/10.1126/science.aar7709} {\bibfield
  {journal} {\bibinfo  {journal} {Science}\ }\textbf {\bibinfo {volume}
  {363}},\ \bibinfo {pages} {eaar7709} (\bibinfo {year} {2019})}\BibitemShut
  {NoStop}%
\bibitem [{\citenamefont {Heiss}\ and\ \citenamefont
  {Sannino}(1990)}]{Heiss1990}%
  \BibitemOpen
  \bibfield  {author} {\bibinfo {author} {\bibfnamefont {W.~D.}\ \bibnamefont
  {Heiss}}\ and\ \bibinfo {author} {\bibfnamefont {A.~L.}\ \bibnamefont
  {Sannino}},\ }\bibfield  {title} {\bibinfo {title} {Avoided level crossing
  and exceptional points},\ }\href {https://doi.org/10.1088/0305-4470/23/7/022}
  {\bibfield  {journal} {\bibinfo  {journal} {J. Phys. A: Math. Gen.}\ }\textbf
  {\bibinfo {volume} {23}},\ \bibinfo {pages} {1167} (\bibinfo {year}
  {1990})}\BibitemShut {NoStop}%
\bibitem [{\citenamefont {Peng}\ \emph {et~al.}(2015)\citenamefont {Peng},
  \citenamefont {Zhou}, \citenamefont {Wei}, \citenamefont {Cui}, \citenamefont
  {Du},\ and\ \citenamefont {Liu}}]{Peng2015}%
  \BibitemOpen
  \bibfield  {author} {\bibinfo {author} {\bibfnamefont {X.}~\bibnamefont
  {Peng}}, \bibinfo {author} {\bibfnamefont {H.}~\bibnamefont {Zhou}}, \bibinfo
  {author} {\bibfnamefont {B.-B.}\ \bibnamefont {Wei}}, \bibinfo {author}
  {\bibfnamefont {J.}~\bibnamefont {Cui}}, \bibinfo {author} {\bibfnamefont
  {J.}~\bibnamefont {Du}},\ and\ \bibinfo {author} {\bibfnamefont {R.-B.}\
  \bibnamefont {Liu}},\ }\bibfield  {title} {\bibinfo {title} {Experimental
  observation of lee-yang zeros},\ }\href
  {https://doi.org/10.1103/PhysRevLett.114.010601} {\bibfield  {journal}
  {\bibinfo  {journal} {Phys. Rev. Lett.}\ }\textbf {\bibinfo {volume} {114}},\
  \bibinfo {pages} {010601} (\bibinfo {year} {2015})}\BibitemShut {NoStop}%
\bibitem [{\citenamefont {Shen}\ and\ \citenamefont {Fu}(2018)}]{Shen2018}%
  \BibitemOpen
  \bibfield  {author} {\bibinfo {author} {\bibfnamefont {H.}~\bibnamefont
  {Shen}}\ and\ \bibinfo {author} {\bibfnamefont {L.}~\bibnamefont {Fu}},\
  }\bibfield  {title} {\bibinfo {title} {Quantum oscillation from in-gap states
  and a non-hermitian landau level problem},\ }\href
  {https://doi.org/10.1103/PhysRevLett.121.026403} {\bibfield  {journal}
  {\bibinfo  {journal} {Phys. Rev. Lett.}\ }\textbf {\bibinfo {volume} {121}},\
  \bibinfo {pages} {026403} (\bibinfo {year} {2018})}\BibitemShut {NoStop}%
\bibitem [{\citenamefont {Kunst}\ \emph {et~al.}(2018)\citenamefont {Kunst},
  \citenamefont {Edvardsson}, \citenamefont {Budich},\ and\ \citenamefont
  {Bergholtz}}]{Kunst2018}%
  \BibitemOpen
  \bibfield  {author} {\bibinfo {author} {\bibfnamefont {F.~K.}\ \bibnamefont
  {Kunst}}, \bibinfo {author} {\bibfnamefont {E.}~\bibnamefont {Edvardsson}},
  \bibinfo {author} {\bibfnamefont {J.~C.}\ \bibnamefont {Budich}},\ and\
  \bibinfo {author} {\bibfnamefont {E.~J.}\ \bibnamefont {Bergholtz}},\
  }\bibfield  {title} {\bibinfo {title} {Biorthogonal bulk-boundary
  correspondence in non-hermitian systems},\ }\href
  {https://doi.org/10.1103/PhysRevLett.121.026808} {\bibfield  {journal}
  {\bibinfo  {journal} {Phys. Rev. Lett.}\ }\textbf {\bibinfo {volume} {121}},\
  \bibinfo {pages} {026808} (\bibinfo {year} {2018})}\BibitemShut {NoStop}%
\bibitem [{\citenamefont {Song}\ \emph
  {et~al.}(2019{\natexlab{a}})\citenamefont {Song}, \citenamefont {Yao},\ and\
  \citenamefont {Wang}}]{Song2019A}%
  \BibitemOpen
  \bibfield  {author} {\bibinfo {author} {\bibfnamefont {F.}~\bibnamefont
  {Song}}, \bibinfo {author} {\bibfnamefont {S.}~\bibnamefont {Yao}},\ and\
  \bibinfo {author} {\bibfnamefont {Z.}~\bibnamefont {Wang}},\ }\bibfield
  {title} {\bibinfo {title} {Non-hermitian skin effect and chiral damping in
  open quantum systems},\ }\href
  {https://doi.org/10.1103/PhysRevLett.123.170401} {\bibfield  {journal}
  {\bibinfo  {journal} {Phys. Rev. Lett.}\ }\textbf {\bibinfo {volume} {123}},\
  \bibinfo {pages} {170401} (\bibinfo {year} {2019}{\natexlab{a}})}\BibitemShut
  {NoStop}%
\bibitem [{\citenamefont {Matsumoto}\ \emph {et~al.}(2020)\citenamefont
  {Matsumoto}, \citenamefont {Kawabata}, \citenamefont {Ashida}, \citenamefont
  {Furukawa},\ and\ \citenamefont {Ueda}}]{Matsumoto2020}%
  \BibitemOpen
  \bibfield  {author} {\bibinfo {author} {\bibfnamefont {N.}~\bibnamefont
  {Matsumoto}}, \bibinfo {author} {\bibfnamefont {K.}~\bibnamefont {Kawabata}},
  \bibinfo {author} {\bibfnamefont {Y.}~\bibnamefont {Ashida}}, \bibinfo
  {author} {\bibfnamefont {S.}~\bibnamefont {Furukawa}},\ and\ \bibinfo
  {author} {\bibfnamefont {M.}~\bibnamefont {Ueda}},\ }\bibfield  {title}
  {\bibinfo {title} {Continuous phase transition without gap closing in
  non-hermitian quantum many-body systems},\ }\href
  {https://doi.org/10.1103/PhysRevLett.125.260601} {\bibfield  {journal}
  {\bibinfo  {journal} {Phys. Rev. Lett.}\ }\textbf {\bibinfo {volume} {125}},\
  \bibinfo {pages} {260601} (\bibinfo {year} {2020})}\BibitemShut {NoStop}%
\bibitem [{\citenamefont {Borgnia}\ \emph {et~al.}(2020)\citenamefont
  {Borgnia}, \citenamefont {Kruchkov},\ and\ \citenamefont
  {Slager}}]{Borgnia2020}%
  \BibitemOpen
  \bibfield  {author} {\bibinfo {author} {\bibfnamefont {D.~S.}\ \bibnamefont
  {Borgnia}}, \bibinfo {author} {\bibfnamefont {A.~J.}\ \bibnamefont
  {Kruchkov}},\ and\ \bibinfo {author} {\bibfnamefont {R.-J.}\ \bibnamefont
  {Slager}},\ }\bibfield  {title} {\bibinfo {title} {Non-hermitian boundary
  modes and topology},\ }\href {https://doi.org/10.1103/PhysRevLett.124.056802}
  {\bibfield  {journal} {\bibinfo  {journal} {Phys. Rev. Lett.}\ }\textbf
  {\bibinfo {volume} {124}},\ \bibinfo {pages} {056802} (\bibinfo {year}
  {2020})}\BibitemShut {NoStop}%
\bibitem [{\citenamefont {Zeuner}\ \emph {et~al.}(2015)\citenamefont {Zeuner},
  \citenamefont {Rechtsman}, \citenamefont {Plotnik}, \citenamefont {Lumer},
  \citenamefont {Nolte}, \citenamefont {Rudner}, \citenamefont {Segev},\ and\
  \citenamefont {Szameit}}]{Zeuner2015}%
  \BibitemOpen
  \bibfield  {author} {\bibinfo {author} {\bibfnamefont {J.~M.}\ \bibnamefont
  {Zeuner}}, \bibinfo {author} {\bibfnamefont {M.~C.}\ \bibnamefont
  {Rechtsman}}, \bibinfo {author} {\bibfnamefont {Y.}~\bibnamefont {Plotnik}},
  \bibinfo {author} {\bibfnamefont {Y.}~\bibnamefont {Lumer}}, \bibinfo
  {author} {\bibfnamefont {S.}~\bibnamefont {Nolte}}, \bibinfo {author}
  {\bibfnamefont {M.~S.}\ \bibnamefont {Rudner}}, \bibinfo {author}
  {\bibfnamefont {M.}~\bibnamefont {Segev}},\ and\ \bibinfo {author}
  {\bibfnamefont {A.}~\bibnamefont {Szameit}},\ }\bibfield  {title} {\bibinfo
  {title} {Observation of a topological transition in the bulk of a
  non-hermitian system},\ }\href
  {https://doi.org/10.1103/PhysRevLett.115.040402} {\bibfield  {journal}
  {\bibinfo  {journal} {Phys. Rev. Lett.}\ }\textbf {\bibinfo {volume} {115}},\
  \bibinfo {pages} {040402} (\bibinfo {year} {2015})}\BibitemShut {NoStop}%
\bibitem [{\citenamefont {Zeng}\ \emph {et~al.}(2020)\citenamefont {Zeng},
  \citenamefont {Yang},\ and\ \citenamefont {Xu}}]{Zeng2020}%
  \BibitemOpen
  \bibfield  {author} {\bibinfo {author} {\bibfnamefont {Q.-B.}\ \bibnamefont
  {Zeng}}, \bibinfo {author} {\bibfnamefont {Y.-B.}\ \bibnamefont {Yang}},\
  and\ \bibinfo {author} {\bibfnamefont {Y.}~\bibnamefont {Xu}},\ }\bibfield
  {title} {\bibinfo {title} {Topological phases in non-hermitian
  aubry-andr\'e-harper models},\ }\href
  {https://doi.org/10.1103/PhysRevB.101.020201} {\bibfield  {journal} {\bibinfo
   {journal} {Phys. Rev. B}\ }\textbf {\bibinfo {volume} {101}},\ \bibinfo
  {pages} {020201} (\bibinfo {year} {2020})}\BibitemShut {NoStop}%
\bibitem [{\citenamefont {Xiong}(2018)}]{Xiong2018}%
  \BibitemOpen
  \bibfield  {author} {\bibinfo {author} {\bibfnamefont {Y.}~\bibnamefont
  {Xiong}},\ }\bibfield  {title} {\bibinfo {title} {Why does bulk boundary
  correspondence fail in some non-hermitian topological models},\ }\href
  {https://doi.org/10.1088/2399-6528/aab64a} {\bibfield  {journal} {\bibinfo
  {journal} {J. Phys. Commun.}\ }\textbf {\bibinfo {volume} {2}},\ \bibinfo
  {pages} {035043} (\bibinfo {year} {2018})}\BibitemShut {NoStop}%
\bibitem [{\citenamefont {Yao}\ and\ \citenamefont {Wang}(2018)}]{Yao2018A}%
  \BibitemOpen
  \bibfield  {author} {\bibinfo {author} {\bibfnamefont {S.}~\bibnamefont
  {Yao}}\ and\ \bibinfo {author} {\bibfnamefont {Z.}~\bibnamefont {Wang}},\
  }\bibfield  {title} {\bibinfo {title} {Edge states and topological invariants
  of non-hermitian systems},\ }\href
  {https://doi.org/10.1103/PhysRevLett.121.086803} {\bibfield  {journal}
  {\bibinfo  {journal} {Phys. Rev. Lett.}\ }\textbf {\bibinfo {volume} {121}},\
  \bibinfo {pages} {086803} (\bibinfo {year} {2018})}\BibitemShut {NoStop}%
\bibitem [{\citenamefont {Yao}\ \emph {et~al.}(2018)\citenamefont {Yao},
  \citenamefont {Song},\ and\ \citenamefont {Wang}}]{Yao2018B}%
  \BibitemOpen
  \bibfield  {author} {\bibinfo {author} {\bibfnamefont {S.}~\bibnamefont
  {Yao}}, \bibinfo {author} {\bibfnamefont {F.}~\bibnamefont {Song}},\ and\
  \bibinfo {author} {\bibfnamefont {Z.}~\bibnamefont {Wang}},\ }\bibfield
  {title} {\bibinfo {title} {Non-hermitian chern bands},\ }\href
  {https://doi.org/10.1103/PhysRevLett.121.136802} {\bibfield  {journal}
  {\bibinfo  {journal} {Phys. Rev. Lett.}\ }\textbf {\bibinfo {volume} {121}},\
  \bibinfo {pages} {136802} (\bibinfo {year} {2018})}\BibitemShut {NoStop}%
\bibitem [{\citenamefont {Song}\ \emph
  {et~al.}(2019{\natexlab{b}})\citenamefont {Song}, \citenamefont {Yao},\ and\
  \citenamefont {Wang}}]{Song2019B}%
  \BibitemOpen
  \bibfield  {author} {\bibinfo {author} {\bibfnamefont {F.}~\bibnamefont
  {Song}}, \bibinfo {author} {\bibfnamefont {S.}~\bibnamefont {Yao}},\ and\
  \bibinfo {author} {\bibfnamefont {Z.}~\bibnamefont {Wang}},\ }\bibfield
  {title} {\bibinfo {title} {Non-hermitian topological invariants in real
  space},\ }\href {https://doi.org/10.1103/PhysRevLett.123.246801} {\bibfield
  {journal} {\bibinfo  {journal} {Phys. Rev. Lett.}\ }\textbf {\bibinfo
  {volume} {123}},\ \bibinfo {pages} {246801} (\bibinfo {year}
  {2019}{\natexlab{b}})}\BibitemShut {NoStop}%
\bibitem [{\citenamefont {Altland}\ and\ \citenamefont
  {Zirnbauer}(1997)}]{Altland1997}%
  \BibitemOpen
  \bibfield  {author} {\bibinfo {author} {\bibfnamefont {A.}~\bibnamefont
  {Altland}}\ and\ \bibinfo {author} {\bibfnamefont {M.~R.}\ \bibnamefont
  {Zirnbauer}},\ }\bibfield  {title} {\bibinfo {title} {Nonstandard symmetry
  classes in mesoscopic normal-superconducting hybrid structures},\ }\href
  {https://doi.org/10.1103/PhysRevB.55.1142} {\bibfield  {journal} {\bibinfo
  {journal} {Phys. Rev. B}\ }\textbf {\bibinfo {volume} {55}},\ \bibinfo
  {pages} {1142} (\bibinfo {year} {1997})}\BibitemShut {NoStop}%
\bibitem [{\citenamefont {Gong}\ \emph {et~al.}(2018)\citenamefont {Gong},
  \citenamefont {Ashida}, \citenamefont {Kawabata}, \citenamefont {Takasan},
  \citenamefont {Higashikawa},\ and\ \citenamefont {Ueda}}]{Gong2018}%
  \BibitemOpen
  \bibfield  {author} {\bibinfo {author} {\bibfnamefont {Z.}~\bibnamefont
  {Gong}}, \bibinfo {author} {\bibfnamefont {Y.}~\bibnamefont {Ashida}},
  \bibinfo {author} {\bibfnamefont {K.}~\bibnamefont {Kawabata}}, \bibinfo
  {author} {\bibfnamefont {K.}~\bibnamefont {Takasan}}, \bibinfo {author}
  {\bibfnamefont {S.}~\bibnamefont {Higashikawa}},\ and\ \bibinfo {author}
  {\bibfnamefont {M.}~\bibnamefont {Ueda}},\ }\bibfield  {title} {\bibinfo
  {title} {Topological phases of non-hermitian systems},\ }\href
  {https://doi.org/10.1103/PhysRevX.8.031079} {\bibfield  {journal} {\bibinfo
  {journal} {Phys. Rev. X}\ }\textbf {\bibinfo {volume} {8}},\ \bibinfo {pages}
  {031079} (\bibinfo {year} {2018})}\BibitemShut {NoStop}%
\bibitem [{\citenamefont {Kawabata}\ \emph
  {et~al.}(2019{\natexlab{a}})\citenamefont {Kawabata}, \citenamefont
  {Shiozaki}, \citenamefont {Ueda},\ and\ \citenamefont
  {Sato}}]{Kawabata2019A}%
  \BibitemOpen
  \bibfield  {author} {\bibinfo {author} {\bibfnamefont {K.}~\bibnamefont
  {Kawabata}}, \bibinfo {author} {\bibfnamefont {K.}~\bibnamefont {Shiozaki}},
  \bibinfo {author} {\bibfnamefont {M.}~\bibnamefont {Ueda}},\ and\ \bibinfo
  {author} {\bibfnamefont {M.}~\bibnamefont {Sato}},\ }\bibfield  {title}
  {\bibinfo {title} {Symmetry and topology in non-hermitian physics},\ }\href
  {https://doi.org/10.1103/PhysRevX.9.041015} {\bibfield  {journal} {\bibinfo
  {journal} {Phys. Rev. X}\ }\textbf {\bibinfo {volume} {9}},\ \bibinfo {pages}
  {041015} (\bibinfo {year} {2019}{\natexlab{a}})}\BibitemShut {NoStop}%
\bibitem [{\citenamefont {Levin}\ and\ \citenamefont {Wen}(2006)}]{Levin2006}%
  \BibitemOpen
  \bibfield  {author} {\bibinfo {author} {\bibfnamefont {M.}~\bibnamefont
  {Levin}}\ and\ \bibinfo {author} {\bibfnamefont {X.-G.}\ \bibnamefont
  {Wen}},\ }\bibfield  {title} {\bibinfo {title} {Detecting topological order
  in a ground state wave function},\ }\href
  {https://doi.org/10.1103/PhysRevLett.96.110405} {\bibfield  {journal}
  {\bibinfo  {journal} {Phys. Rev. Lett.}\ }\textbf {\bibinfo {volume} {96}},\
  \bibinfo {pages} {110405} (\bibinfo {year} {2006})}\BibitemShut {NoStop}%
\bibitem [{\citenamefont {Kitaev}\ and\ \citenamefont
  {Preskill}(2006)}]{Kitaev2006}%
  \BibitemOpen
  \bibfield  {author} {\bibinfo {author} {\bibfnamefont {A.}~\bibnamefont
  {Kitaev}}\ and\ \bibinfo {author} {\bibfnamefont {J.}~\bibnamefont
  {Preskill}},\ }\bibfield  {title} {\bibinfo {title} {Topological entanglement
  entropy},\ }\href {https://doi.org/10.1103/PhysRevLett.96.110404} {\bibfield
  {journal} {\bibinfo  {journal} {Phys. Rev. Lett.}\ }\textbf {\bibinfo
  {volume} {96}},\ \bibinfo {pages} {110404} (\bibinfo {year}
  {2006})}\BibitemShut {NoStop}%
\bibitem [{\citenamefont {Chen}\ \emph {et~al.}(2010)\citenamefont {Chen},
  \citenamefont {Gu},\ and\ \citenamefont {Wen}}]{Chen2010}%
  \BibitemOpen
  \bibfield  {author} {\bibinfo {author} {\bibfnamefont {X.}~\bibnamefont
  {Chen}}, \bibinfo {author} {\bibfnamefont {Z.-C.}\ \bibnamefont {Gu}},\ and\
  \bibinfo {author} {\bibfnamefont {X.-G.}\ \bibnamefont {Wen}},\ }\bibfield
  {title} {\bibinfo {title} {Local unitary transformation, long-range quantum
  entanglement, wave function renormalization, and topological order},\ }\href
  {https://doi.org/10.1103/PhysRevB.82.155138} {\bibfield  {journal} {\bibinfo
  {journal} {Phys. Rev. B}\ }\textbf {\bibinfo {volume} {82}},\ \bibinfo
  {pages} {155138} (\bibinfo {year} {2010})}\BibitemShut {NoStop}%
\bibitem [{\citenamefont {Landau}\ and\ \citenamefont
  {Ginzburg}(1950)}]{Landau1950}%
  \BibitemOpen
  \bibfield  {author} {\bibinfo {author} {\bibfnamefont {L.~D.}\ \bibnamefont
  {Landau}}\ and\ \bibinfo {author} {\bibfnamefont {V.~L.}\ \bibnamefont
  {Ginzburg}},\ }\bibfield  {title} {\bibinfo {title} {{On the theory of
  superconductivity}},\ }\href {https://cds.cern.ch/record/486430} {\bibfield
  {journal} {\bibinfo  {journal} {Zh. Eksp. Teor. Fiz.}\ }\textbf {\bibinfo
  {volume} {20}},\ \bibinfo {pages} {1064} (\bibinfo {year}
  {1950})}\BibitemShut {NoStop}%
\bibitem [{\citenamefont {Gu}\ and\ \citenamefont {Wen}(2009)}]{Gu2009}%
  \BibitemOpen
  \bibfield  {author} {\bibinfo {author} {\bibfnamefont {Z.-C.}\ \bibnamefont
  {Gu}}\ and\ \bibinfo {author} {\bibfnamefont {X.-G.}\ \bibnamefont {Wen}},\
  }\bibfield  {title} {\bibinfo {title} {Tensor-entanglement-filtering
  renormalization approach and symmetry-protected topological order},\ }\href
  {https://doi.org/10.1103/PhysRevB.80.155131} {\bibfield  {journal} {\bibinfo
  {journal} {Phys. Rev. B}\ }\textbf {\bibinfo {volume} {80}},\ \bibinfo
  {pages} {155131} (\bibinfo {year} {2009})}\BibitemShut {NoStop}%
\bibitem [{\citenamefont {Chen}\ \emph
  {et~al.}(2011{\natexlab{a}})\citenamefont {Chen}, \citenamefont {Gu},\ and\
  \citenamefont {Wen}}]{Chen2011A}%
  \BibitemOpen
  \bibfield  {author} {\bibinfo {author} {\bibfnamefont {X.}~\bibnamefont
  {Chen}}, \bibinfo {author} {\bibfnamefont {Z.-C.}\ \bibnamefont {Gu}},\ and\
  \bibinfo {author} {\bibfnamefont {X.-G.}\ \bibnamefont {Wen}},\ }\bibfield
  {title} {\bibinfo {title} {Classification of gapped symmetric phases in
  one-dimensional spin systems},\ }\href
  {https://doi.org/10.1103/PhysRevB.83.035107} {\bibfield  {journal} {\bibinfo
  {journal} {Phys. Rev. B}\ }\textbf {\bibinfo {volume} {83}},\ \bibinfo
  {pages} {035107} (\bibinfo {year} {2011}{\natexlab{a}})}\BibitemShut
  {NoStop}%
\bibitem [{\citenamefont {Chen}\ \emph {et~al.}(2013)\citenamefont {Chen},
  \citenamefont {Gu}, \citenamefont {Liu},\ and\ \citenamefont
  {Wen}}]{Chen2013}%
  \BibitemOpen
  \bibfield  {author} {\bibinfo {author} {\bibfnamefont {X.}~\bibnamefont
  {Chen}}, \bibinfo {author} {\bibfnamefont {Z.-C.}\ \bibnamefont {Gu}},
  \bibinfo {author} {\bibfnamefont {Z.-X.}\ \bibnamefont {Liu}},\ and\ \bibinfo
  {author} {\bibfnamefont {X.-G.}\ \bibnamefont {Wen}},\ }\bibfield  {title}
  {\bibinfo {title} {Symmetry protected topological orders and the group
  cohomology of their symmetry group},\ }\href
  {https://doi.org/10.1103/PhysRevB.87.155114} {\bibfield  {journal} {\bibinfo
  {journal} {Phys. Rev. B}\ }\textbf {\bibinfo {volume} {87}},\ \bibinfo
  {pages} {155114} (\bibinfo {year} {2013})}\BibitemShut {NoStop}%
\bibitem [{\citenamefont {Chen}\ \emph {et~al.}(2023)\citenamefont {Chen},
  \citenamefont {Song},\ and\ \citenamefont {Lado}}]{Chen2023}%
  \BibitemOpen
  \bibfield  {author} {\bibinfo {author} {\bibfnamefont {G.}~\bibnamefont
  {Chen}}, \bibinfo {author} {\bibfnamefont {F.}~\bibnamefont {Song}},\ and\
  \bibinfo {author} {\bibfnamefont {J.~L.}\ \bibnamefont {Lado}},\ }\bibfield
  {title} {\bibinfo {title} {Topological spin excitations in non-hermitian spin
  chains with a generalized kernel polynomial algorithm},\ }\href
  {https://doi.org/10.1103/PhysRevLett.130.100401} {\bibfield  {journal}
  {\bibinfo  {journal} {Phys. Rev. Lett.}\ }\textbf {\bibinfo {volume} {130}},\
  \bibinfo {pages} {100401} (\bibinfo {year} {2023})}\BibitemShut {NoStop}%
\bibitem [{\citenamefont {Zhang}\ \emph {et~al.}(2020)\citenamefont {Zhang},
  \citenamefont {Jin},\ and\ \citenamefont {Song}}]{Zhang2020}%
  \BibitemOpen
  \bibfield  {author} {\bibinfo {author} {\bibfnamefont {X.~Z.}\ \bibnamefont
  {Zhang}}, \bibinfo {author} {\bibfnamefont {L.}~\bibnamefont {Jin}},\ and\
  \bibinfo {author} {\bibfnamefont {Z.}~\bibnamefont {Song}},\ }\bibfield
  {title} {\bibinfo {title} {Dynamic magnetization in non-hermitian quantum
  spin systems},\ }\href {https://doi.org/10.1103/PhysRevB.101.224301}
  {\bibfield  {journal} {\bibinfo  {journal} {Phys. Rev. B}\ }\textbf {\bibinfo
  {volume} {101}},\ \bibinfo {pages} {224301} (\bibinfo {year}
  {2020})}\BibitemShut {NoStop}%
\bibitem [{\citenamefont {Turkeshi}\ and\ \citenamefont
  {Schir\'o}(2023)}]{Turkeshi2023}%
  \BibitemOpen
  \bibfield  {author} {\bibinfo {author} {\bibfnamefont {X.}~\bibnamefont
  {Turkeshi}}\ and\ \bibinfo {author} {\bibfnamefont {M.}~\bibnamefont
  {Schir\'o}},\ }\bibfield  {title} {\bibinfo {title} {Entanglement and
  correlation spreading in non-hermitian spin chains},\ }\href
  {https://doi.org/10.1103/PhysRevB.107.L020403} {\bibfield  {journal}
  {\bibinfo  {journal} {Phys. Rev. B}\ }\textbf {\bibinfo {volume} {107}},\
  \bibinfo {pages} {L020403} (\bibinfo {year} {2023})}\BibitemShut {NoStop}%
\bibitem [{\citenamefont {Castro-Alvaredo}\ and\ \citenamefont
  {Fring}(2009)}]{CastroAlvaredo2009}%
  \BibitemOpen
  \bibfield  {author} {\bibinfo {author} {\bibfnamefont {O.~A.}\ \bibnamefont
  {Castro-Alvaredo}}\ and\ \bibinfo {author} {\bibfnamefont {A.}~\bibnamefont
  {Fring}},\ }\bibfield  {title} {\bibinfo {title} {A spin chain model with
  non-hermitian interaction: the ising quantum spin chain in an imaginary
  field},\ }\href {https://doi.org/10.1088/1751-8113/42/46/465211} {\bibfield
  {journal} {\bibinfo  {journal} {Journal of Physics A: Mathematical and
  Theoretical}\ }\textbf {\bibinfo {volume} {42}},\ \bibinfo {pages} {465211}
  (\bibinfo {year} {2009})}\BibitemShut {NoStop}%
\bibitem [{\citenamefont {Mu}\ \emph {et~al.}(2020)\citenamefont {Mu},
  \citenamefont {Lee}, \citenamefont {Li},\ and\ \citenamefont
  {Gong}}]{Mu2020}%
  \BibitemOpen
  \bibfield  {author} {\bibinfo {author} {\bibfnamefont {S.}~\bibnamefont
  {Mu}}, \bibinfo {author} {\bibfnamefont {C.~H.}\ \bibnamefont {Lee}},
  \bibinfo {author} {\bibfnamefont {L.}~\bibnamefont {Li}},\ and\ \bibinfo
  {author} {\bibfnamefont {J.}~\bibnamefont {Gong}},\ }\bibfield  {title}
  {\bibinfo {title} {Emergent fermi surface in a many-body non-hermitian
  fermionic chain},\ }\href {https://doi.org/10.1103/PhysRevB.102.081115}
  {\bibfield  {journal} {\bibinfo  {journal} {Phys. Rev. B}\ }\textbf {\bibinfo
  {volume} {102}},\ \bibinfo {pages} {081115} (\bibinfo {year}
  {2020})}\BibitemShut {NoStop}%
\bibitem [{\citenamefont {Alsallom}\ \emph {et~al.}(2022)\citenamefont
  {Alsallom}, \citenamefont {Herviou}, \citenamefont {Yazyev},\ and\
  \citenamefont {Brzezi\ifmmode~\acute{n}\else \'{n}\fi{}ska}}]{Alsallom2022}%
  \BibitemOpen
  \bibfield  {author} {\bibinfo {author} {\bibfnamefont {F.}~\bibnamefont
  {Alsallom}}, \bibinfo {author} {\bibfnamefont {L.}~\bibnamefont {Herviou}},
  \bibinfo {author} {\bibfnamefont {O.~V.}\ \bibnamefont {Yazyev}},\ and\
  \bibinfo {author} {\bibfnamefont {M.}~\bibnamefont
  {Brzezi\ifmmode~\acute{n}\else \'{n}\fi{}ska}},\ }\bibfield  {title}
  {\bibinfo {title} {Fate of the non-hermitian skin effect in many-body
  fermionic systems},\ }\href
  {https://doi.org/10.1103/PhysRevResearch.4.033122} {\bibfield  {journal}
  {\bibinfo  {journal} {Phys. Rev. Res.}\ }\textbf {\bibinfo {volume} {4}},\
  \bibinfo {pages} {033122} (\bibinfo {year} {2022})}\BibitemShut {NoStop}%
\bibitem [{\citenamefont {Kawabata}\ \emph {et~al.}(2023)\citenamefont
  {Kawabata}, \citenamefont {Numasawa},\ and\ \citenamefont
  {Ryu}}]{Kawabata2023}%
  \BibitemOpen
  \bibfield  {author} {\bibinfo {author} {\bibfnamefont {K.}~\bibnamefont
  {Kawabata}}, \bibinfo {author} {\bibfnamefont {T.}~\bibnamefont {Numasawa}},\
  and\ \bibinfo {author} {\bibfnamefont {S.}~\bibnamefont {Ryu}},\ }\bibfield
  {title} {\bibinfo {title} {Entanglement phase transition induced by the
  non-hermitian skin effect},\ }\href
  {https://doi.org/10.1103/PhysRevX.13.021007} {\bibfield  {journal} {\bibinfo
  {journal} {Phys. Rev. X}\ }\textbf {\bibinfo {volume} {13}},\ \bibinfo
  {pages} {021007} (\bibinfo {year} {2023})}\BibitemShut {NoStop}%
\bibitem [{\citenamefont {Hamazaki}\ \emph {et~al.}(2019)\citenamefont
  {Hamazaki}, \citenamefont {Kawabata},\ and\ \citenamefont
  {Ueda}}]{Hamazaki2019}%
  \BibitemOpen
  \bibfield  {author} {\bibinfo {author} {\bibfnamefont {R.}~\bibnamefont
  {Hamazaki}}, \bibinfo {author} {\bibfnamefont {K.}~\bibnamefont {Kawabata}},\
  and\ \bibinfo {author} {\bibfnamefont {M.}~\bibnamefont {Ueda}},\ }\bibfield
  {title} {\bibinfo {title} {Non-hermitian many-body localization},\ }\href
  {https://doi.org/10.1103/PhysRevLett.123.090603} {\bibfield  {journal}
  {\bibinfo  {journal} {Phys. Rev. Lett.}\ }\textbf {\bibinfo {volume} {123}},\
  \bibinfo {pages} {090603} (\bibinfo {year} {2019})}\BibitemShut {NoStop}%
\bibitem [{\citenamefont {Zhai}\ \emph {et~al.}(2020)\citenamefont {Zhai},
  \citenamefont {Yin},\ and\ \citenamefont {Huang}}]{Zhai2020}%
  \BibitemOpen
  \bibfield  {author} {\bibinfo {author} {\bibfnamefont {L.-J.}\ \bibnamefont
  {Zhai}}, \bibinfo {author} {\bibfnamefont {S.}~\bibnamefont {Yin}},\ and\
  \bibinfo {author} {\bibfnamefont {G.-Y.}\ \bibnamefont {Huang}},\ }\bibfield
  {title} {\bibinfo {title} {Many-body localization in a non-hermitian
  quasiperiodic system},\ }\href {https://doi.org/10.1103/PhysRevB.102.064206}
  {\bibfield  {journal} {\bibinfo  {journal} {Phys. Rev. B}\ }\textbf {\bibinfo
  {volume} {102}},\ \bibinfo {pages} {064206} (\bibinfo {year}
  {2020})}\BibitemShut {NoStop}%
\bibitem [{\citenamefont {Wang}\ \emph {et~al.}(2023)\citenamefont {Wang},
  \citenamefont {Suthar}, \citenamefont {Jen}, \citenamefont {Hsu},\ and\
  \citenamefont {You}}]{Wang2023}%
  \BibitemOpen
  \bibfield  {author} {\bibinfo {author} {\bibfnamefont {Y.-C.}\ \bibnamefont
  {Wang}}, \bibinfo {author} {\bibfnamefont {K.}~\bibnamefont {Suthar}},
  \bibinfo {author} {\bibfnamefont {H.~H.}\ \bibnamefont {Jen}}, \bibinfo
  {author} {\bibfnamefont {Y.-T.}\ \bibnamefont {Hsu}},\ and\ \bibinfo {author}
  {\bibfnamefont {J.-S.}\ \bibnamefont {You}},\ }\href@noop {} {\bibinfo
  {title} {Non-hermitian skin effects on many-body localized and thermal
  phases}} (\bibinfo {year} {2023}),\ \Eprint
  {https://arxiv.org/abs/2210.12998} {arXiv:2210.12998} \BibitemShut {NoStop}%
\bibitem [{\citenamefont {Shen}\ \emph {et~al.}(2023)\citenamefont {Shen},
  \citenamefont {Guo},\ and\ \citenamefont {Yang}}]{Shen2023}%
  \BibitemOpen
  \bibfield  {author} {\bibinfo {author} {\bibfnamefont {R.}~\bibnamefont
  {Shen}}, \bibinfo {author} {\bibfnamefont {Y.}~\bibnamefont {Guo}},\ and\
  \bibinfo {author} {\bibfnamefont {S.}~\bibnamefont {Yang}},\ }\bibfield
  {title} {\bibinfo {title} {Construction of non-hermitian parent hamiltonian
  from matrix product states},\ }\href
  {https://doi.org/10.1103/PhysRevLett.130.220401} {\bibfield  {journal}
  {\bibinfo  {journal} {Phys. Rev. Lett.}\ }\textbf {\bibinfo {volume} {130}},\
  \bibinfo {pages} {220401} (\bibinfo {year} {2023})}\BibitemShut {NoStop}%
\bibitem [{\citenamefont {Hastings}(2009)}]{Hastings2009}%
  \BibitemOpen
  \bibfield  {author} {\bibinfo {author} {\bibfnamefont {M.~B.}\ \bibnamefont
  {Hastings}},\ }\bibfield  {title} {\bibinfo {title} {Light-cone matrix
  product},\ }\href {https://doi.org/10.1063/1.3149556} {\bibfield  {journal}
  {\bibinfo  {journal} {J. Math. Phys.}\ }\textbf {\bibinfo {volume} {50}},\
  \bibinfo {pages} {095207} (\bibinfo {year} {2009})}\BibitemShut {NoStop}%
\bibitem [{\citenamefont {Li}\ and\ \citenamefont {Haldane}(2008)}]{Li2008}%
  \BibitemOpen
  \bibfield  {author} {\bibinfo {author} {\bibfnamefont {H.}~\bibnamefont
  {Li}}\ and\ \bibinfo {author} {\bibfnamefont {F.~D.~M.}\ \bibnamefont
  {Haldane}},\ }\bibfield  {title} {\bibinfo {title} {Entanglement spectrum as
  a generalization of entanglement entropy: Identification of topological order
  in non-abelian fractional quantum hall effect states},\ }\href
  {https://doi.org/10.1103/PhysRevLett.101.010504} {\bibfield  {journal}
  {\bibinfo  {journal} {Phys. Rev. Lett.}\ }\textbf {\bibinfo {volume} {101}},\
  \bibinfo {pages} {010504} (\bibinfo {year} {2008})}\BibitemShut {NoStop}%
\bibitem [{\citenamefont {Pollmann}\ \emph {et~al.}(2010)\citenamefont
  {Pollmann}, \citenamefont {Turner}, \citenamefont {Berg},\ and\ \citenamefont
  {Oshikawa}}]{Pollmann2010}%
  \BibitemOpen
  \bibfield  {author} {\bibinfo {author} {\bibfnamefont {F.}~\bibnamefont
  {Pollmann}}, \bibinfo {author} {\bibfnamefont {A.~M.}\ \bibnamefont
  {Turner}}, \bibinfo {author} {\bibfnamefont {E.}~\bibnamefont {Berg}},\ and\
  \bibinfo {author} {\bibfnamefont {M.}~\bibnamefont {Oshikawa}},\ }\bibfield
  {title} {\bibinfo {title} {Entanglement spectrum of a topological phase in
  one dimension},\ }\href {https://doi.org/10.1103/PhysRevB.81.064439}
  {\bibfield  {journal} {\bibinfo  {journal} {Phys. Rev. B}\ }\textbf {\bibinfo
  {volume} {81}},\ \bibinfo {pages} {064439} (\bibinfo {year}
  {2010})}\BibitemShut {NoStop}%
\bibitem [{\citenamefont {Xi}\ \emph {et~al.}(2021)\citenamefont {Xi},
  \citenamefont {Zhang}, \citenamefont {Gu},\ and\ \citenamefont
  {Chen}}]{Xi2021}%
  \BibitemOpen
  \bibfield  {author} {\bibinfo {author} {\bibfnamefont {W.}~\bibnamefont
  {Xi}}, \bibinfo {author} {\bibfnamefont {Z.-H.}\ \bibnamefont {Zhang}},
  \bibinfo {author} {\bibfnamefont {Z.-C.}\ \bibnamefont {Gu}},\ and\ \bibinfo
  {author} {\bibfnamefont {W.-Q.}\ \bibnamefont {Chen}},\ }\bibfield  {title}
  {\bibinfo {title} {Classification of topological phases in one dimensional
  interacting non-hermitian systems and emergent unitarity},\ }\href
  {https://doi.org/https://doi.org/10.1016/j.scib.2021.04.027} {\bibfield
  {journal} {\bibinfo  {journal} {Science Bulletin}\ }\textbf {\bibinfo
  {volume} {66}},\ \bibinfo {pages} {1731} (\bibinfo {year}
  {2021})}\BibitemShut {NoStop}%
\bibitem [{\citenamefont {Brody}(2013)}]{Brody2013}%
  \BibitemOpen
  \bibfield  {author} {\bibinfo {author} {\bibfnamefont {D.~C.}\ \bibnamefont
  {Brody}},\ }\bibfield  {title} {\bibinfo {title} {Biorthogonal quantum
  mechanics},\ }\href@noop {} {\bibfield  {journal} {\bibinfo  {journal} {J.
  Phys. A-Math. Theor.}\ }\textbf {\bibinfo {volume} {47}},\ \bibinfo {pages}
  {035305} (\bibinfo {year} {2013})}\BibitemShut {NoStop}%
\bibitem [{\citenamefont {Lee}\ \emph {et~al.}(2020)\citenamefont {Lee},
  \citenamefont {Lee},\ and\ \citenamefont {Yang}}]{Lee2020}%
  \BibitemOpen
  \bibfield  {author} {\bibinfo {author} {\bibfnamefont {E.}~\bibnamefont
  {Lee}}, \bibinfo {author} {\bibfnamefont {H.}~\bibnamefont {Lee}},\ and\
  \bibinfo {author} {\bibfnamefont {B.-J.}\ \bibnamefont {Yang}},\ }\bibfield
  {title} {\bibinfo {title} {Many-body approach to non-hermitian physics in
  fermionic systems},\ }\href {https://doi.org/10.1103/PhysRevB.101.121109}
  {\bibfield  {journal} {\bibinfo  {journal} {Phys. Rev. B}\ }\textbf {\bibinfo
  {volume} {101}},\ \bibinfo {pages} {121109} (\bibinfo {year}
  {2020})}\BibitemShut {NoStop}%
\bibitem [{\citenamefont {Grimaudo}\ \emph {et~al.}(2020)\citenamefont
  {Grimaudo}, \citenamefont {Messina}, \citenamefont {Sergi}, \citenamefont
  {Vitanov},\ and\ \citenamefont {Filippov}}]{Grimaudo2020}%
  \BibitemOpen
  \bibfield  {author} {\bibinfo {author} {\bibfnamefont {R.}~\bibnamefont
  {Grimaudo}}, \bibinfo {author} {\bibfnamefont {A.}~\bibnamefont {Messina}},
  \bibinfo {author} {\bibfnamefont {A.}~\bibnamefont {Sergi}}, \bibinfo
  {author} {\bibfnamefont {N.~V.}\ \bibnamefont {Vitanov}},\ and\ \bibinfo
  {author} {\bibfnamefont {S.~N.}\ \bibnamefont {Filippov}},\ }\bibfield
  {title} {\bibinfo {title} {Two-qubit entanglement generation through
  non-hermitian hamiltonians induced by repeated measurements on an ancilla},\
  }\href {https://doi.org/10.3390/e22101184} {\bibfield  {journal} {\bibinfo
  {journal} {Entropy}\ }\textbf {\bibinfo {volume} {22}},\ \bibinfo {pages}
  {1184} (\bibinfo {year} {2020})}\BibitemShut {NoStop}%
\bibitem [{\citenamefont {Grimaldi}\ \emph {et~al.}(2021)\citenamefont
  {Grimaldi}, \citenamefont {Sergi},\ and\ \citenamefont
  {Messina}}]{Grimaldi2021}%
  \BibitemOpen
  \bibfield  {author} {\bibinfo {author} {\bibfnamefont {A.}~\bibnamefont
  {Grimaldi}}, \bibinfo {author} {\bibfnamefont {A.}~\bibnamefont {Sergi}},\
  and\ \bibinfo {author} {\bibfnamefont {A.}~\bibnamefont {Messina}},\
  }\bibfield  {title} {\bibinfo {title} {Evolution of a non-hermitian quantum
  single-molecule junction at constant temperature},\ }\href
  {https://doi.org/10.3390/e23020147} {\bibfield  {journal} {\bibinfo
  {journal} {Entropy}\ }\textbf {\bibinfo {volume} {23}},\ \bibinfo {pages}
  {147} (\bibinfo {year} {2021})}\BibitemShut {NoStop}%
\bibitem [{\citenamefont {Herviou}\ \emph {et~al.}(2019)\citenamefont
  {Herviou}, \citenamefont {Regnault},\ and\ \citenamefont
  {Bardarson}}]{Herviou2019}%
  \BibitemOpen
  \bibfield  {author} {\bibinfo {author} {\bibfnamefont {L.}~\bibnamefont
  {Herviou}}, \bibinfo {author} {\bibfnamefont {N.}~\bibnamefont {Regnault}},\
  and\ \bibinfo {author} {\bibfnamefont {J.~H.}\ \bibnamefont {Bardarson}},\
  }\bibfield  {title} {\bibinfo {title} {Entanglement spectrum and symmetries
  in non-hermitian fermionic non-interacting models},\ }\href
  {https://doi.org/10.21468/SciPostPhys.7.5.069} {\bibfield  {journal}
  {\bibinfo  {journal} {SciPost Phys.}\ }\textbf {\bibinfo {volume} {7}},\
  \bibinfo {pages} {069} (\bibinfo {year} {2019})}\BibitemShut {NoStop}%
\bibitem [{\citenamefont {Ju}\ \emph {et~al.}(2019)\citenamefont {Ju},
  \citenamefont {Miranowicz}, \citenamefont {Chen},\ and\ \citenamefont
  {Nori}}]{Ju2019}%
  \BibitemOpen
  \bibfield  {author} {\bibinfo {author} {\bibfnamefont {C.-Y.}\ \bibnamefont
  {Ju}}, \bibinfo {author} {\bibfnamefont {A.}~\bibnamefont {Miranowicz}},
  \bibinfo {author} {\bibfnamefont {G.-Y.}\ \bibnamefont {Chen}},\ and\
  \bibinfo {author} {\bibfnamefont {F.}~\bibnamefont {Nori}},\ }\bibfield
  {title} {\bibinfo {title} {Non-hermitian hamiltonians and no-go theorems in
  quantum information},\ }\href {https://doi.org/10.1103/PhysRevA.100.062118}
  {\bibfield  {journal} {\bibinfo  {journal} {Phys. Rev. A}\ }\textbf {\bibinfo
  {volume} {100}},\ \bibinfo {pages} {062118} (\bibinfo {year}
  {2019})}\BibitemShut {NoStop}%
\bibitem [{\citenamefont {Kawabata}\ \emph
  {et~al.}(2019{\natexlab{b}})\citenamefont {Kawabata}, \citenamefont
  {Bessho},\ and\ \citenamefont {Sato}}]{Kawabata2019B}%
  \BibitemOpen
  \bibfield  {author} {\bibinfo {author} {\bibfnamefont {K.}~\bibnamefont
  {Kawabata}}, \bibinfo {author} {\bibfnamefont {T.}~\bibnamefont {Bessho}},\
  and\ \bibinfo {author} {\bibfnamefont {M.}~\bibnamefont {Sato}},\ }\bibfield
  {title} {\bibinfo {title} {Classification of exceptional points and
  non-hermitian topological semimetals},\ }\href
  {https://doi.org/10.1103/PhysRevLett.123.066405} {\bibfield  {journal}
  {\bibinfo  {journal} {Phys. Rev. Lett.}\ }\textbf {\bibinfo {volume} {123}},\
  \bibinfo {pages} {066405} (\bibinfo {year} {2019}{\natexlab{b}})}\BibitemShut
  {NoStop}%
\bibitem [{\citenamefont {P\'erez-Garc\'{\i}a}\ \emph
  {et~al.}(2007)\citenamefont {P\'erez-Garc\'{\i}a}, \citenamefont
  {Verstraete}, \citenamefont {Wolf},\ and\ \citenamefont
  {Cirac}}]{PerezGarcia2007}%
  \BibitemOpen
  \bibfield  {author} {\bibinfo {author} {\bibfnamefont {D.}~\bibnamefont
  {P\'erez-Garc\'{\i}a}}, \bibinfo {author} {\bibfnamefont {F.}~\bibnamefont
  {Verstraete}}, \bibinfo {author} {\bibfnamefont {M.~M.}\ \bibnamefont
  {Wolf}},\ and\ \bibinfo {author} {\bibfnamefont {J.~I.}\ \bibnamefont
  {Cirac}},\ }\bibfield  {title} {\bibinfo {title} {Matrix product state
  representations},\ }\href {https://doi.org/10.26421/QIC7.5-6-1} {\bibfield
  {journal} {\bibinfo  {journal} {Quantum Info. Comput.}\ }\textbf {\bibinfo
  {volume} {7}},\ \bibinfo {pages} {401–430} (\bibinfo {year}
  {2007})}\BibitemShut {NoStop}%
\bibitem [{\citenamefont {P\'erez-Garc\'{\i}a}\ \emph
  {et~al.}(2008)\citenamefont {P\'erez-Garc\'{\i}a}, \citenamefont
  {Verstraete}, \citenamefont {Wolf},\ and\ \citenamefont
  {Cirac}}]{PerezGarcia2008}%
  \BibitemOpen
  \bibfield  {author} {\bibinfo {author} {\bibfnamefont {D.}~\bibnamefont
  {P\'erez-Garc\'{\i}a}}, \bibinfo {author} {\bibfnamefont {F.}~\bibnamefont
  {Verstraete}}, \bibinfo {author} {\bibfnamefont {M.~M.}\ \bibnamefont
  {Wolf}},\ and\ \bibinfo {author} {\bibfnamefont {J.~I.}\ \bibnamefont
  {Cirac}},\ }\bibfield  {title} {\bibinfo {title} {Peps as unique ground
  states of local hamiltonians},\ }\href {https://doi.org/10.26421/QIC8.6-7-6}
  {\bibfield  {journal} {\bibinfo  {journal} {Quantum Info. Comput.}\ }\textbf
  {\bibinfo {volume} {8}},\ \bibinfo {pages} {650–663} (\bibinfo {year}
  {2008})}\BibitemShut {NoStop}%
\bibitem [{\citenamefont {Liu}\ \emph {et~al.}(2011)\citenamefont {Liu},
  \citenamefont {Liu},\ and\ \citenamefont {Wen}}]{Liu2011}%
  \BibitemOpen
  \bibfield  {author} {\bibinfo {author} {\bibfnamefont {Z.-X.}\ \bibnamefont
  {Liu}}, \bibinfo {author} {\bibfnamefont {M.}~\bibnamefont {Liu}},\ and\
  \bibinfo {author} {\bibfnamefont {X.-G.}\ \bibnamefont {Wen}},\ }\bibfield
  {title} {\bibinfo {title} {Gapped quantum phases for the $s=1$ spin chain
  with ${D}_{2h}$ symmetry},\ }\href
  {https://doi.org/10.1103/PhysRevB.84.075135} {\bibfield  {journal} {\bibinfo
  {journal} {Phys. Rev. B}\ }\textbf {\bibinfo {volume} {84}},\ \bibinfo
  {pages} {075135} (\bibinfo {year} {2011})}\BibitemShut {NoStop}%
\bibitem [{\citenamefont {Chen}\ \emph
  {et~al.}(2011{\natexlab{b}})\citenamefont {Chen}, \citenamefont {Gu},\ and\
  \citenamefont {Wen}}]{Chen2011B}%
  \BibitemOpen
  \bibfield  {author} {\bibinfo {author} {\bibfnamefont {X.}~\bibnamefont
  {Chen}}, \bibinfo {author} {\bibfnamefont {Z.-C.}\ \bibnamefont {Gu}},\ and\
  \bibinfo {author} {\bibfnamefont {X.-G.}\ \bibnamefont {Wen}},\ }\bibfield
  {title} {\bibinfo {title} {Complete classification of one-dimensional gapped
  quantum phases in interacting spin systems},\ }\href
  {https://doi.org/10.1103/PhysRevB.84.235128} {\bibfield  {journal} {\bibinfo
  {journal} {Phys. Rev. B}\ }\textbf {\bibinfo {volume} {84}},\ \bibinfo
  {pages} {235128} (\bibinfo {year} {2011}{\natexlab{b}})}\BibitemShut
  {NoStop}%
\bibitem [{\citenamefont {Pollmann}\ and\ \citenamefont
  {Turner}(2012)}]{Pollmann2012}%
  \BibitemOpen
  \bibfield  {author} {\bibinfo {author} {\bibfnamefont {F.}~\bibnamefont
  {Pollmann}}\ and\ \bibinfo {author} {\bibfnamefont {A.~M.}\ \bibnamefont
  {Turner}},\ }\bibfield  {title} {\bibinfo {title} {Detection of
  symmetry-protected topological phases in one dimension},\ }\href
  {https://doi.org/10.1103/PhysRevB.86.125441} {\bibfield  {journal} {\bibinfo
  {journal} {Phys. Rev. B}\ }\textbf {\bibinfo {volume} {86}},\ \bibinfo
  {pages} {125441} (\bibinfo {year} {2012})}\BibitemShut {NoStop}%
\bibitem [{\citenamefont {Zeng}\ \emph {et~al.}(2019)\citenamefont {Zeng},
  \citenamefont {Chen}, \citenamefont {Zhou},\ and\ \citenamefont
  {Wen}}]{Zeng2019}%
  \BibitemOpen
  \bibfield  {author} {\bibinfo {author} {\bibfnamefont {B.}~\bibnamefont
  {Zeng}}, \bibinfo {author} {\bibfnamefont {X.}~\bibnamefont {Chen}}, \bibinfo
  {author} {\bibfnamefont {D.-L.}\ \bibnamefont {Zhou}},\ and\ \bibinfo
  {author} {\bibfnamefont {X.-G.}\ \bibnamefont {Wen}},\ }\href
  {https://doi.org/10.1007/978-1-4939-9084-9} {\emph {\bibinfo {title} {Quantum
  Information Meets Quantum Matter}}}\ (\bibinfo  {publisher} {Springer New
  York},\ \bibinfo {year} {2019})\BibitemShut {NoStop}%
\bibitem [{\citenamefont {Cao}\ \emph {et~al.}(2023)\citenamefont {Cao},
  \citenamefont {Li}, \citenamefont {Zhao}, \citenamefont {Guo}, \citenamefont
  {Qi}, \citenamefont {Chang}, \citenamefont {Zhou}, \citenamefont {Xu},\ and\
  \citenamefont {Duan}}]{Cao2023}%
  \BibitemOpen
  \bibfield  {author} {\bibinfo {author} {\bibfnamefont {M.-M.}\ \bibnamefont
  {Cao}}, \bibinfo {author} {\bibfnamefont {K.}~\bibnamefont {Li}}, \bibinfo
  {author} {\bibfnamefont {W.-D.}\ \bibnamefont {Zhao}}, \bibinfo {author}
  {\bibfnamefont {W.-X.}\ \bibnamefont {Guo}}, \bibinfo {author} {\bibfnamefont
  {B.-X.}\ \bibnamefont {Qi}}, \bibinfo {author} {\bibfnamefont {X.-Y.}\
  \bibnamefont {Chang}}, \bibinfo {author} {\bibfnamefont {Z.-C.}\ \bibnamefont
  {Zhou}}, \bibinfo {author} {\bibfnamefont {Y.}~\bibnamefont {Xu}},\ and\
  \bibinfo {author} {\bibfnamefont {L.-M.}\ \bibnamefont {Duan}},\ }\bibfield
  {title} {\bibinfo {title} {Probing complex-energy topology via non-hermitian
  absorption spectroscopy in a trapped ion simulator},\ }\href
  {https://doi.org/10.1103/PhysRevLett.130.163001} {\bibfield  {journal}
  {\bibinfo  {journal} {Phys. Rev. Lett.}\ }\textbf {\bibinfo {volume} {130}},\
  \bibinfo {pages} {163001} (\bibinfo {year} {2023})}\BibitemShut {NoStop}%
\end{thebibliography}%

\end{document}